\documentclass[useAMS,usenatbib,12pt]{article}
\usepackage{natbib} 
\usepackage{fancyhdr}
\usepackage[figuresright]{rotating} 
\usepackage{adjustbox} 
\usepackage{amsmath,amssymb,amsthm,color,mathrsfs}
\usepackage{multirow} 
\usepackage{booktabs}
\usepackage{hyperref}
\usepackage{authblk} 
\hypersetup{
    colorlinks=true,
    linkcolor=blue,
    filecolor=blue,      
    urlcolor=blue,
    citecolor=blue,
}
\usepackage{setspace}
\usepackage{geometry}
\usepackage{ulem,enumitem,anysize}

\newcommand{\rom}[1]{\uppercase\expandafter{\romannumeral #1\relax}}
\newcommand{\indep}{\rotatebox[origin=c]{90}{$\models$}}

\newtheorem{lemma}{Lemma}

\newtheorem{proposition}{Proposition}
\newtheorem{assumption}{Assumption}

\def\bSig\mathbf{\Sigma}
\usepackage{sectsty}
\sectionfont{\large}
\subsectionfont{\normalsize}
\subsubsectionfont{\normalsize}

\title{\large\vspace{-2.0cm}\textbf{ Matched Design for Marginal Causal Effect on Restricted Mean Survival Time in Observational Studies}}
\author[1]{Zihan Lin}
\author[1]{Ai Ni}
\author[1]{Bo Lu\thanks{Corresponding author: Bo Lu (lu.232@osu.edu)}}
\date{}
\affil[1]{Division of Biostatistics, College of Public Health,
The Ohio State University, Columbus, OH, USA}

\begin{document}
\maketitle
\vspace{-1.0cm}
\begin{abstract}

\setstretch{1.5}
Investigating the causal relationship between exposure and the time-to-event outcome is an important topic in biomedical research. Previous literature has discussed the potential issues of using the hazard ratio as a marginal causal effect measure due to its noncollapsibility property. In this paper, we advocate using the restricted mean survival time (RMST) difference as the marginal causal effect measure, which is collapsible and has a simple interpretation as the difference of area under survival curves over a certain time horizon. To address both measured and unmeasured confounding, a matched design with sensitivity analysis is proposed. Matching is used to pair similar treated and untreated subjects together, which is more robust to outcome model misspecification. Our propensity score matched RMST difference estimator is shown to be asymptotically unbiased and the corresponding variance estimator is calculated by accounting for the correlation due to matching. The simulation study also demonstrates that our method has adequate empirical performance and outperforms many competing methods used in practice. To assess the impact of unmeasured confounding, we develop a sensitivity analysis strategy by adapting the E-value approach to matched data. We apply the proposed method to the Atherosclerosis Risk in Communities Study (ARIC) to examine the causal effect of smoking on stroke-free survival.

\end{abstract}


\section{Introduction}
\label{s:intro}

\subsection{Causal Inference for Observational Survival Data}
In biomedical studies, time-to-event is a commonly used outcome measure and the statistical analyses for such data are usually referred to as survival analysis. Investigating the causal relationship between exposure and the time-to-event outcome is an important topic, with either randomized trials or observational studies. Causal inference for observational survival data has several challenges. First, since not all subjects can be observed for the full duration of time to event, the survival data suffer from censoring, which is a type of missing data problem. Therefore, standard statistical methods are usually not sufficient to handle both censoring and the missingness of potential outcomes. Second, the hazard ratio is a popular choice for measuring the association of survival outcomes between two groups, for convenience and easy interpretation. However, the hazard ratio is generally not an appropriate marginal causal effect measure due to its noncollapsibility property \citep{greenland1999confounding,tmcop,sjolander2016note}. Other effect measures need to be considered to warrant valid marginal causal interpretation for survival data. Third, confounding is a major challenge in observational studies, which includes both measured and unmeasured confounders. 
Propensity score adjustments are popular tools for controlling the observed confounding \citep{cps}. But even with successful adjustment of observed confounding, observational data are still vulnerable to unmeasured confounding. Thus, appropriate sensitivity analysis needs to be developed to assess the impact of hidden bias \citep{rosenbaum2020design}. 


The issues of using the hazard ratio as a marginal causal effect measure have been discussed extensively in the literature. \citet{greenland1999confounding} pointed out that the hazard ratio has the noncollapsibility property when the treatment effect is nonzero.  \citet{hoh} argued that using the hazard ratio as a treatment effect measure may not have valid causal interpretation even in randomized studies, since the hazard ratio has a built-in selection bias and may change over time. \citet{tmcop} studied the estimation of treatment effect in the presence of confounders and found the amount of confounding due to noncollapsibility in the Cox proportional hazards (PH) model would be very difficult to quantify. \citet{doescox} offered a more theoretical perspective on the conditions under which the hazard has a valid causal interpretation.  They suggested that the hazard function $h(t,x,z)$ must satisfy an additive assumption $h(t,x,z)=a(t,z)+b(t,x)$ to yield a causal interpretation, where $a(t,z)$ is a function of survival time $t$ and treatment assignment $z$ and $b(t,x)$ is a function of survival time $t$ and covariates $x$. \citet{ni2021stratified} further illustrated that even under a PH model, the marginal hazard ratio is not a constant, after integrating out covariates. Thus, a valid and simple-to-use causal effect measure for survival outcome is highly desirable. 


\subsection{RMST Difference as A Marginal Causal Effect Measure}

The restricted mean survival time (RMST) has been used in randomized clinical studies to evaluate treatment effects \citep{royston2013restricted,trinquart2016comparison}. The RMST difference is more advantageous than the hazard ratio as a marginal effect measure. First, the RMST has an intuitive interpretation as the area under the survival curve over a certain time horizon. Second, the RMST difference is the difference of truncated mean survival time between two groups, which is essentially a mean difference. So it is collapsible, meaning that the marginal and conditional effects are compatible. Third, the treatment effect measured by the RMST difference can be asymptotically unbiasedly estimated without PH assumption, while the conventional Cox model heavily relies on such assumption.

To take advantage of the collapsibility of RMST difference, we can construct RMST regression by including covariates to better control for confounding or increase estimation efficiency. Several methods of regressing RMST on multiple covariates have been developed. \citet{karrison1987restricted} examined the RMST as an index for comparing survival in two groups and proposed to model the hazard with piece-wise exponential models assuming covariates have a multiplicative effect on the hazard. \citet{zucker1998restricted} further simplified the implementation procedure for Karrison's method and provided an extended version to achieve robustness against model misspecification. \citet{andersen2004regression} compared several regression analysis methods of mean survival time and RMST, and they proposed a regression method based on pseudo-observations. \citet{tian2014predicting} developed an RMST regression model with adjustment for baseline covariates. They constructed an estimating equation with the inverse probability of censoring weighting (IPCW) to obtain consistent estimates.
\citet{wang2018modeling} models the RMST using generalized estimating equation methods, which allows censoring to depend on both baseline covariates and time-dependent factors. 

Though RMST differences have been reported in many randomized clinical studies, there is only limited discussion of using RMST in observational studies, probably due to the challenge of confounding adjustment. Propensity score weighting and stratification methods have been used in the literature, but not propensity score matching. \citet{zhang2012double} derived a double-robust estimator for RMST difference based on the inverse probability of treatment weighted (IPTW) estimating equation with augmentation term. To adjust for confounding factors, they built three working models for survival time, treatment assignment, and censoring, then incorporated them into the augmentation term. They assumed the PH assumption in outcome modeling, which might be violated in practice. \citet{conner2019adjusted} proposed a weighted method to compare the adjusted RMST difference directly. Unlike Zhang and Schaubel's work, Conner et al. estimated the RMST based on the Kaplan-Meier(KM) estimator rather than the Nelson-Aalen estimator. They adjusted the KM estimator with IPTW and derived the adjusted RMST by integrating the IPTW-adjusted KM estimator. \citet{ni2021stratified} proposed a propensity score stratified RMST difference estimation strategy to examine the marginal causal effect with observational survival data, which can combine stratification with further regression adjustment.

\subsection{A Motivating Example: Atherosclerosis Risk in Communities (ARIC)}

In the United States, stroke is a severe disease that causes serious disability for adults and is a leading cause of death \citep{kochanek2014mortality,members2016heart}. Several previous studies have shown that smoking is an important risk factor for stroke \citep{wolf1988cigarette,shinton1989meta}, and even passive smoking could increase the risk of stroke \citep{bonita1999passive}.
Although the causal pathway between smoking and stroke is unclear, \citet{shah2010smoking} found that the more people smoke the more likely they were to have a stroke, and people who quit smoking showed a significantly lower risk of stroke, which provides some evidence for the causal relationship between smoking and stroke.  

The Atherosclerosis Risk in Communities (ARIC) Study \citep{aric1989atherosclerosis} is a prospective cohort study conducted in four U.S. communities. Four thousand adults aged 45–64 years old were randomly sampled from each of four U.S. communities, and the final dataset contains information of 15,792 individuals. After a baseline examination during 1987 to 1989, subjects were followed up for the development of incident ischemic stroke, and first definite or probable hospitalized stroke. Due to the length of follow-up, not all event times were observed, so the data were subjected to censoring. One primary outcome is the time to first stroke or death (whichever comes first), and a subject is censored if the incidence of stroke or death is not observed by the end of the study. We try to answer a causal question, using this ARIC dataset: how smokers' stroke-free survival would change had they not smoked at baseline. Matched design is a natural choice to address this as it is about the causal effect of those being exposed to smoking, rather than for the entire population. 

Existing literature mostly used the Cox PH model to analyze the ARIC data. \citet{kwon2016association} and \citet{ding2019cigarette} studied the association between smoking status and risk of stroke, using the Cox PH model to estimate the HR of smoking status on the risk of stroke. Thus, the estimated effects were interpreted as conditional rather than marginal. Moreover, it is possible that some prognostic factors were not included in the confounder adjusted regression model, which would lead to biased estimates of conditional effects. An analysis with RMST as the effect measure may provide new insight into this research.

In this paper, we propose a propensity score matching based RMST difference estimator and develop a corresponding sensitivity analysis strategy for assessing the impact due to unmeasured confounding. We apply this method to the ARIC study to examine the causal effect of smoking on stroke-free survival. The rest of the paper is organized as follows: In section \ref{s:method}, we set up the notation and assumptions, and describe the proposed RMST estimator with its theoretical properties. In section \ref{s:simu}, we conduct a simulation study to examine the empirical performance of our proposed method under different scenarios and also compare it with several commonly used methods in practice. In section \ref{s:sensit}, we develop a sensitivity analysis strategy by adapting the E-value approach to matched data. In section \ref{s:real}, we present the analysis results of the ARIC data. Section \ref{s:discuss} concludes the paper with some discussions.

\section{Method: Matched RMST Difference Estimation}
\label{s:method}
\subsection{Notation and Assumptions}
We follow the potential outcomes framework proposed by \citet{rubin1974estimating} to define the causal effects. In a two-arm survival analysis study, let $A$ be the treatment assignment indicator (or more generally, the exposure status), such that $A=1$ indicates being exposed to the treatment and $A=0$ indicates being exposed to the control. Let $T^A$ denote the potential event time and $S^A(t)$ denote the corresponding survival function for a subject under treatment $A$. 
The following two assumptions are extensions of commonly used assumptions for causal inference in observational studies \citep{imbens2015causal}. 

\begin{assumption}
\label{sutva}
Stable Unit Treatment Value Assumption (SUTVA). The potential survival times for one individual in the population do not vary with the treatment assigned to others. And there are no different versions of the specified treatment level. 
\end{assumption}
\begin{assumption}
\label{strong_ignore}
Treatment assignment is strongly ignorable given covariates $X$, that is $(T^0,T^1)\indep A|X$ and $0<pr(A=1|X)<1$.
\end{assumption}

The potential restricted event time is defined as $Z^A=\min (T^A,\tau)$, where $\tau$ is the truncation time point, which is usually pre-specified at the design stage based on clinical relevance and study feasibility.  Both $T^A$ and $Z^A$ are subject to censoring by a random variable $C$. We introduce two additional assumptions for survival data.


\begin{assumption}
\label{indep_censor}
The censoring random variable $C$ is independent of $T^A$ and $Z^A$ given treatment indicator $A$ and covariates $X$, that is $C\indep T^A |(A,X)$ and $C\indep Z^A |(A,X)$.
\end{assumption}

\begin{assumption}
\label{tau_assump}
The truncation time point is smaller than the largest observed survival time, $\tau<t_{\max}$, where $t_{\max}$ is the largest follow up time (event or censored). 
\end{assumption}

Assumption \ref{tau_assump} is a technical one to ensure that the pre-specified $\tau$ is clinically meaningful and RMST can be asymptotically unbiasedly estimated. 

Let $\delta^A=I(Z^A<C)$ denote the censoring indicator, then the observed restricted time is defined as $Y^A=\min(Z^A,C)=(Z^A)^{\delta^A}C^{(1-\delta^A)}$. For a subject under treatment $A$, the potential outcome of restricted mean survival time is defined as $\mu^A(\tau)=E(Z^A)=\int^\tau_0S^A(t)dt$, then the average treatment effect (ATE) on RMST, denoted by $\Delta_{ATE}$, can be defined as
$$\Delta_{ATE}=\mu^1(\tau)-\mu^0(\tau)=E(Z^1)-E(Z^0)=\int^\tau_0[S^1(t)-S^0(t)]dt.$$

Since $Z^A=\min (T^A, \tau)$ and $\tau$ is a fixed constant,  $(T^0,T^1)\indep A|X$ implies  $(Z^{0},Z^{1})\indep\ A|X$. Following Theorem 3 in \citet{cps}, we can establish the strongly ignorability based on propensity score $e(X)=P(A=1|X)$ for survival outcomes in proposition \ref{note_prop} (proof provided in Web Appendix A).

\begin{proposition}
\label{note_prop}
Given assumptions \ref{sutva}-\ref{strong_ignore}, we have $(T^0,T^1)\indep A|e(X)$, which further implies $(Z^0,Z^1)\indep A|e(X)$.
\end{proposition}


\subsection{Matched RMST Difference Estimator}
In randomized trials, the marginal causal effect of treatment on RMST can be asymptotically unbiasedly estimated \citep{fleming2011counting} by direct contrast of group-specific RMST estimates since confounding effects are eliminated by design. In observational studies, however, additional adjustments are needed for confounding control. Propensity score based approaches are popular for this purpose, which may take the form of matching, stratification, or weighting \citep{cps, bang2005doubly}. Among different propensity score adjustment strategies, matching is a design tool that selects comparable control units to match with treated units and it often results in more robust causal effect estimates as it does not rely on outcome model specification. Usually, matching uses all treated and a subset of control units, so it estimates the average treatment effect on the treated (ATT) \citep{imbens2015causal}. 

Our proposed propensity score matched RMST estimation includes the following steps: \\
(1) \textit{Propensity Score Estimation} \\
The propensity score is defined as the conditional probability of treatment given a vector of observed covariates \citep{cps}. We estimate the propensity score by fitting a logistic regression on $A$ with $X$, though other estimation options, either parametric or nonparametric, are also available \citep{mccaffrey2004propensity,westreich2010propensity}.\\
(2) \textit{Propensity Score Matching} \\
We use the optimal matching algorithm by \citet{Roptmatch} to create pair matches without replacement based on the estimated propensity score, the unmatched controls will be removed from the matched sample. Matching quality is assessed by checking the post-matching covariate balance. Any substantial covariate imbalance would lead to a recalibration of the propensity score model. We will proceed to the next step only after a satisfactory balance is achieved.  \\
(3)\textit{Treatment Effect Estimation}\\
Suppose we obtain $n$ pairs of data through matching, where each pair contains exactly one treated and one control subject. We estimate the RMST, $\mu(\tau)$, by $\hat{\mu}(\tau)=\int^\tau_0\hat{S}(t)d(t)$, where $\hat{S}(t)$ is estimated by the nonparametric KM method. Let $\hat{S}^0(t)$ and $\hat{S}^1(t)$ denote the KM estimates of survival function for control and treated groups in the matched sample, respectively. Based on the matched sample, our estimator for averaged treatment effect on the treated (ATT) is 
$$\hat{\Delta}_{ATT}=\hat{\mu}^1(\tau)-\hat{\mu}^0(\tau)=\int^\tau_0[\hat{S}^1(t)-\hat{S}^0(t)]dt.$$


The following two propositions show that the matched RMST difference estimator is asymptotically unbiased (both proofs are provided in Web Appendix A).

\begin{proposition}
\label{RMST_unbias}
Given assumptions \ref{sutva}-\ref{tau_assump}, the RMST estimator based on KM method given propensity score $e(X)$ and treatment group $A$, denoted as $\hat{\mu}_{e(X),A}$, is an asymptotically unbiased estimator for $\mu_{e(X),A}$ given $\tau<t_{\max}$.
\end{proposition}


\begin{proposition}
\label{match_rmst_unbias}
Given assumptions \ref{sutva}-\ref{tau_assump},  $\hat{\Delta}_{ATT}$ is asymptotically unbiased.
\end{proposition}

\subsection{Variance Estimation}
The matching process may introduce correlation between the two subjects in the same pair, as they are matched on similar propensity scores. Therefore, the variance calculation of $\hat{\Delta}_{ATT}$ needs to account for such correlation: 
$$var(\hat{\Delta}_{ATT})=var[\int^\tau_0\hat{S}^0(t_0)dt_0]+var[\int^\tau_0\hat{S}^1(t_1)dt_1]-2cov[\int^\tau_0\hat{S}^0(t_0)dt_0,\int^\tau_0\hat{S}^1(t_1)dt_1].$$
The overall variance has two components, the marginal variance of RMST estimates and their covariance. For two dependent event times with independent censoring and no competing risk, \citet{murray2000variance} provided closed-form asymptotic covariance formulas for KM survival estimates and corresponding RMST estimates. To address the dependence structure introduced in the matching process, we adapt their formulas to compute the covariance between control and treated group RMST estimates in the matched sample. 

Specifically, let $T_0$ be the event time for a subject from the control group with marginal hazard function $h_0(\cdot)$, and $T_1$ be the event time for a subject from the treatment group with marginal hazard function $h_1(\cdot)$, then the event times for a matched pair of control and treated can be denoted as $(T_0, T_1)$. Let $C$ be the censoring variable, then the observed time can be denoted as $\tilde{T}_0=\min (T_0, C)$ for control group with censoring indicator $\delta_0=I(T_0<C)$, and $\tilde{T}_1=\min (T_1, C)$ for treated group with censoring indicator $\delta_1=I(T_1<C)$. Then, the joint hazard function is 
$h_{ij}(u,v)=\underset{\Delta u,\Delta v\rightarrow 0}{\lim}\frac{1}{\Delta u\Delta v}P(u\leq \tilde{T}_i<u+\Delta u, v\leq \tilde{T}_j<v+\Delta v,\delta_i=1,\delta_j=1|\tilde{T}_i\geq u, \tilde{T}_j\geq v)$ where $i,j\in\{0,1\}$, and the conditional hazard function is $h_{i|j}(u|v)=\underset{\Delta u\rightarrow 0}{\lim}\frac{1}{\Delta u}P(u\leq \tilde{T}_i<u+\Delta u,\delta_i=1|\tilde{T}_i\geq u, \tilde{T}_j\geq v)$ where $i,j\in\{0,1\}$. Then, the covariance between two RMSTs can be computed as
\begin{align} 
&cov[\int^\tau_0\hat{S}^0(t_0)dt_0,\int^\tau_0\hat{S}^1(t_1)dt_1] \nonumber \\
=&\frac{1}{n}\int^\tau_0\int^\tau_0\hat{S}^0(t_0)\hat{S}^1(t_1)\int^{t_0}_0\int^{t_1}_0G_{01}(u,v)dvdudt_0dt_1 \nonumber \\
=&\frac{1}{n}\int^\tau_0\int^\tau_0\left[\int^\tau_v\hat{S}^0(t)dt\right]\left[\int^\tau_u\hat{S}^1(t)dt\right]G_{01}(u,v)dvdu \nonumber
\end{align}
where 
$G_{01}(u,v)=\frac{P(\tilde{T}_0\geq u, \tilde{T}_1\geq v)}{P(\tilde{T}_0\geq u)P(\tilde{T}_1\geq v)}[h_{01}(u,v)-h_{0|1}(u|v)h_1(v)-h_{1|0}(v|u)h_0(u)+h_0(u)h_1(v)]$. Details about the computation of function $G_{01}(u,v)$ are included in Web Appendix C.\\
For the marginal variances, two methods may be considered:
 \begin{enumerate}
 \item Murray's Method: the above covariance formulas can be used to compute the marginal variance, since the marginal variance of RMST could be written as the covariance with itself, that is $var(\int^\tau_0\hat{S}(t)dt)=cov[\int^\tau_0\hat{S}(t)dt,\int^\tau_0\hat{S}(t)dt]$. 
 \item Hosmer's Method: we may also consider the computation method introduced in \citet{hosmer2011applied}.
Let $t_1<t_2<\cdots<t_D$ represent distinct event times. For each $k=1,\cdots,D$, let $Y_k$ be the number of surviving units just prior to event time $t_k$, and let $d_k$ be the number of events at $t_k$.
Let $\hat{S}(t_k)=\overset{k}{\underset{l=1}{\prod}}(1-\frac{d_l}{Y_l})$ denotes the KM estimate of the survival function at event time $t_k$, and let $N_\tau$ be the number of $t_k$ values that are less than truncation time point $\tau$, then the RMST is estimated by
$$\int^\tau_0\hat{S}(t)dt=\overset{N_\tau}{\underset{k=1}{\sum}}\hat{S}(t_{k-1})(t_k-t_{k-1})+\hat{S}(t_{N_\tau})(\tau-t_{N_\tau})$$ 
and the marginal variance of RMST can be estimated as
$$var(\int^\tau_0\hat{S}(t)dt)=\frac{m}{m-1}\overset{N_\tau}{\underset{k=1}{\sum}}\frac{d_kA^2_k}{Y_k(Y_k-d_k)}$$
where
$A_k=\int^\tau_{t_k}\hat{S}(t)dt=\overset{N_\tau}{\underset{l=k}{\sum}}\hat{S}(t_l)(t_{l+1}-t_l)+\hat{S}(t_{N_\tau})(\tau-t_{N_\tau})$ and $m=\overset{N_\tau}{\underset{l=1}{\sum}}d_l$.
 \end{enumerate}
In our simulation studies, we present the variance estimates under Murray's method since the results from these two methods turn out to be very close. 

\section{Simulation Studies}
\label{s:simu}
\subsection{Data Generation}
To assess the empirical performance of the proposed method, we simulate an observational dataset with known confounders. Several existing methods for causal inference with survival outcomes are compared.

We generate ten independent baseline covariates denoted by $X_1$ to $X_{10}$. Among them, $X_1, X_3,\cdots, X_9$ are five binary covariates following Bernoulli distribution with parameters 0.2, 0.4, 0.6, 0.8, 0.5, respectively, and $X_2,X_4,\cdots, X_{10}$ are five continuous covariates following standard normal distribution. We then generate potential survival time  $T^1$ as the outcome under treatment and potential survival time $T^0$ as the outcome under control from Weibull distribution \citep{bender2005generating}. Specifically, we simulate a uniform random variable $U$ on [0,1] then generate the potential survival time as below,
$$T^j=(-\frac{log(U)}{\lambda_{0j}\exp(\beta_Aj+X_1+1.2X_4+1.4X_6+1.6X_7+1.6X_8+1.4X_9+1.2X_{10})})^{\frac{1}{\nu_j}}$$
where $A$ is the treatment indicator and $\beta_A$ is the conditional multiplicative treatment effect on the hazard function given covariates, and $\nu_j$ and $\lambda_{0j}$ are the shape and baseline scale parameters of Weibull distribution for treatment group $j$, respectively. 
When $\nu_0=\nu_1$, we have proportional hazards model, otherwise the model is non-proportional hazards. The treatment indicator $A$ is generated from Bernoulli distribution with $P(A=1)$ defined by the logistic model $\text{logit}(P(A=1))=-1.95+\log(1.2)X_1+\log(1.1)X_2+\log(1.4)X_3+\log(1.2)X_4+\log(1.6)X_5+\log(1.3)X_6+\log(1.8)X_7$. Thus, $X_1, X_4, X_6$ and $X_7$ are true confounders. This setup allows about 20$\%$ of the population to be exposed to treatment.

The censoring variable $C$ is generated from an exponential distribution with rate parameter $\theta$, where $\theta=\gamma\exp(0.2X_4+0.1X_7)$ and $\gamma$ is the baseline rate parameter. The true event time is $T=T^0(1-A)+T^1A$.
Let $\tau$ be the pre-specified truncation time point, we generate the restricted event time $Z=\min(T,\tau)$ and the observed restricted time $Y=\min (Z,C)=\min (T,C,\tau)$. The restricted event time $Z$ is censored if the observed time $C<Z$ with censoring status $\delta_Z=I(Z<C)$, otherwise it is non-censored. 

We simulate 500 datasets of sample size 2500 for each scenario and set the truncation time point $\tau$ to 100. The true RMST difference is determined by calculating the empirical difference between the potential RMSTs under treated and control conditions, and we compute both ATT and ATE versions of true RMST difference to serve as benchmarks for different methods as appropriate. In the $j^{th}$ simulated dataset, we calculate $\Delta_j=\sum_{i=1}^{n}\frac{Z_i^1-Z_i^0}{n}$, where $Z_i^A=\min(T_i^A,\tau)$ is the potential restricted event time for the $i$th individual and $n$ is the sample size of the treated group (for ATT) or the entire sample (for ATE). Then, the true marginal effect on RMST is calculated as $\Delta_0=\sum_{j}^{500}\frac{\Delta_j}{500}$.

Both proportional hazards (PH) and non-proportional hazards (NP) settings are examined. Under both settings, we set $\beta_A$ to five different values: $0, -0.4, -0.8, -1.2,-2$.  For each treatment effect value, we also consider four different levels of 
censoring rates (CR), which are $0\%, 20\%, 40\%, 60\%$. Detailed parameter setup for PH and NP scenarios in observational studies are summarized in Table \ref{tab:obs_setup}.

\subsection{Estimation Strategies}
The proposed method is compared with three existing estimation strategies:
\begin{enumerate}
\item \textit{Propensity score matched RMST estimation.}
This is our proposed propensity score matched estimation as described in the previous section, and the propensity score is estimated from the true propensity score model. The estimated treatment effect is compared to the ATT version of the true RMST difference in our simulations. 
\item \textit{Conner's IPTW RMST estimation.}\\
This method is proposed by \citet{conner2019adjusted}, and they estimated the RMST based on inverse probability treatment weighting (IPTW) adjusted Kaplan-Meier estimator. In our simulation, we use the ATT version of weight to adjust for observed confounding, so it is compared to the true ATT RMST difference. The propensity score is estimated from the true propensity score model.
\item \textit{RMST regression.}\\
This method is proposed by \citet{tian2014predicting}, which uses the IPCW estimating equation with identity link function to estimate treatment effect on RMST with adjustment for covariates. The estimated treatment effect is compared to ATE version of the true RMST difference. We consider four different outcome models in the RMST regression :  (1) outcome model using the treatment indicator only; (2) outcome model using the true covariate set; (3) outcome model using all covariates; (4) outcome model using a wrong covariate set. Due to space limitations, only results of RMST regression with true covariates are summarized in the following section, which has the best performance among the four models. An important caveat is that the RMST regression model with the true covariate set does not represent the true outcome model since the data are generated based on a hazard model. 
\item \textit{Inverse Probability Treatment Weighting (IPTW) Cox regression.}\\
This method estimates $\beta_A$. The propensity score is estimated from the true propensity score model. We use the ATT weight to fit a weighted Cox regression model and regard $\beta_A$ as the truth to calculate the bias and coverage probabilities since there is no single value true marginal hazard ratio. We consider four different outcome models in the IPTW Cox regression :  (1) outcome model using treatment indicator only; (2) outcome model using the true covariate set; (3) outcome model using all covariates; (4) outcome model using a wrong covariate set. Due to space limitations, only results of IPTW Cox regression with the true covariate set are summarized in the following section, which has the best performance among the four models. We understand that the results here are not directly comparable to the first three methods, as they are based on different effect measures. Due to the high popularity of the IPTW Cox model in practice, however, we think there is some value in presenting the results as a reference.   
\end{enumerate}

\subsection{Performance Assessment}
We summarize treatment effect estimates from 500 Monte Carlo iterations into four measures: 
(1) percentage bias (Bias $\%$), which is the bias divided by the true value for nonzero treatment effect scenarios. For the zero treatment effect scenario, we just report the bias. The bias is computed as the average of 500 treatment effect estimates minus the truth; (2) coverage probability (CP), which is the proportion of 500 95$\%$ confidence intervals that cover the truth; (3) model-based standard error (SEM), which is the average of the 500 estimated standard errors from the model-based formula; (4) empirical standard error (SEE), which is the standard error of the 500 point estimates of treatment effect.

\subsection{Results}


Simulation results under PH setting are summarized in Table \ref{tab:obsph-1}. The proposed matched RMST method generates unbiased estimates of the target parameters under most scenarios, and the coverage probabilities are around 95$\%$. For a small effect size ($\beta_A=-0.4$), the bias is a bit large for a high censoring rate. Conner's method has a similar performance, with moderately larger biases. Averaging across all scenarios, bias from Conner's method is 65\% higher than our method.  
The results of the IPTW Cox model are mostly good since we use the correct outcome model.  The coverage probability may be a bit lower than the nominal level, sometimes, which may be due to the underestimated standard error.  The RMST regression method shows a relatively large percentage bias and lower coverage probability, especially under scenarios with large treatment effects. This is likely due to the incorrect covariate functional form specification in the model even though we include the right covariate set.

Simulation results under NP setting are summarized in Table \ref{tab:obsnp-1}. Both our matched RMST method and Conner's method have similar performance (with the latter having more bias) as under the PH setting since these methods do not rely on the PH assumption. The RMST regression method performs somewhat worse, with a bigger bias and much lower than ideal coverage probabilities. Because the PH assumption does not hold here, the IPTW Cox model completely misses the target with large bias and very small coverage probabilities. 


\section{Sensitivity Analysis Based on Matched Design}
\label{s:sensit}

\subsection{An Overview of E-value}

Propensity score adjustment can only control for observed confounding. Unmeasured confounding is likely to be present in observational studies since researchers have no control over the treatment assignment. Thus, sensitivity analysis is important to assess the impact of hidden bias.

Ding and VanderWeele \citep{ding2016sensitivity,vanderweele2017sensitivity} developed a new sensitivity analysis strategy, known as the E-value method. It assumes a hypothetical unmeasured confounder, $U$, and provides a lower bound of the strength of association on the risk ratio scale that $U$ would have to have with both the exposure and the outcome, to explain away the
observed association. Below is a brief review of the conventional E-value method to set the stage for our sensitivity analysis of RMST difference. 

Let $E$ denote a binary exposure and $D$ denote a binary outcome, $e(X)$ is a vector of measured confounders and $U$ is a binary unmeasured confounder with levels $k=0,1$. The observed relative risk of exposure $E$ on the outcome $D$ within stratum of $e(X)=e(x)$ is  
$$\text{RR}_{ED|e(x)}^{obs}=\frac{P(D=1|E=1, e(X)=e(x))}{P(D=1|E=0, e(X)=e(x))}.$$
Then the relative risk of exposure on level $k$ of the unmeasured confounder $U$ within stratum of $e(X)=e(x)$ is 
$$\text{RR}_{EU,k|e(x)}=\frac{P(U=k|E=1,e(X)=e(x))}{P(U=k|E=0, e(X)=e(x))}.$$
Since $U$ is not observed, to facilitate the analysis, we take the maximal relative risk of $E$ on $U$ within stratum $e(X)=e(x)$, denoted as $RR_{EU|e(x)}=\underset{k}{\max}RR_{EU,k|e(x)}$. Similarly, we can define an upper bound for the relative risk between $U$ and $D$ as $\text{RR}_{UD|e(x)}=\max(\text{RR}_{UD|E=0,e(x)},\text{RR}_{UD|E=1,e(x)})$, where $\text{RR}_{UD|E,e(x)}$ is an upper bound of the relative risk between $U$ and $D$ in exposed or unexposed group respectively, within stratum $e(X)=e(x)$. If $e(X)$ and $U$ are sufficient to control for all confounding effects, the true causal relative risk is
$$\text{RR}^{true}_{ED|e(x)}=\frac{\sum^{1}_{k=0}P(D=1|E=1, e(X)=e(x), U=k)P(U=k|e(X)=e(x))}{\sum^{1}_{k=0}P(D=1|E=0, e(X)=e(x), U=k)P(U=k|e(X)=e(x))}.$$
The relative risk pair $(\text{RR}_{EU|e(x)},\text{RR}_{UD|e(x)})$ are used to measure the strength of confounding between the exposure $E$ and the outcome $D$ induced by the confounder $U$ within the stratum of $e(X)=e(x)$. Even though we cannot estimate the true relative risk, its ratio with the observed relative risk is bounded by the following quantity, which is a function of the sensitivity parameters $\text{RR}_{EU|e(x)}$ and $\text{RR}_{UD|e(x)}$.
$$\frac{\text{RR}^{obs}_{ED|e(x)}}{\text{RR}^{true}_{ED|e(x)}}\leq\frac{\text{RR}_{EU|e(x)}\times \text{RR}_{UD|e(x)}}{\text{RR}_{EU|e(x)}+\text{RR}_{UD|e(x)}-1}$$

For given values of $\text{RR}_{EU|e(x)}$ and $\text{RR}_{UD|e(x)}$, we can identify a range of possible values for the true relative risk. If the range covers one, the observed significant association would be explained away by the presence of unmeasured confounding at the given magnitude.  



\subsection{Sensitivity Analysis on RMST difference with matched data}
\label{rmst_sa}
This E-value method can be adapted to conduct sensitivity analysis for our RMST difference estimator in matched design. There are a series of propositions to justify the theoretical validity of using the E-value for the RMST difference estimator. In the interest of space, we just illustrate the main idea in this subsection and leave the propositions and their detailed proofs in Web Appendix B. 

Let $A$ be the treatment indicator and $Z=\min (T,\tau)$ be the RMST outcome, where $T$ is the event time and $\tau$ is the truncation time point. Let $e(X)$ be the propensity score and $U$ be a binary unmeasured confounder with levels $k=0,1$. 
The relative risk of treatment $A$ on level $k$ of the unmeasured confounder $U$ with a given propensity score value $e(X)=e(x)$ is defined as
$$RR_{AU,k|e(x)}=\frac{P(U=k|A=1,e(X)=e(x))}{P(U=k|A=0, e(X)=e(x))}.$$
The maximal relative risk of $A$ on $U$ with $e(X)=e(x)$ is 
$RR_{AU|e(x)}=\underset{k}{\max}RR_{AU,k|e(x)}$.
We define the expectations of the RMST outcome $Z$ given $U=u$ and $e(X)=e(x)$ with and without treatment as 
$$r_1(u)=E(Z|A=1,U=u, e(X)=e(x)),$$
$$r_0(u)=E(Z|A=0,U=u, e(X)=e(x)).$$

Then, the mean ratios of $U$ on $Z$ with and without treatment with $e(X)=e(x)$ are defined as 
$$MR_{UZ|A=1,e(X)=e(x)}=\frac{\underset{u}{\max}r_1(u)}{\underset{u}{\min}r_1(u)},\ MR_{UZ|A=0,e(X)=e(x)}=\frac{\underset{u}{\max}r_0(u)}{\underset{u}{\min}r_0(u)},$$
$$MR_{UZ|e(X)=e(x)}=\max (MR_{UZ|A=1,e(X)=e(x)},MR_{UZ|A=0,e(X)=e(x)}).$$
As shown in proposition 4 in Web Appendix B, both unmeasured confounder parameters $RR_{AU|e(X)=e(x)}$ and $MR_{UZ|e(X)=e(x)}$ are no less than 1. Then we can identify the bounding factor as
\begin{equation}\label{eq:bf}
BF_{U|e(X)=e(x)}=\frac{RR_{AU|e(X)=e(x)}\times MR_{UZ|e(X)=e(x)}}{RR_{AU|e(X)=e(x)}+MR_{UZ|e(X)=e(x)}-1}
\end{equation}
where $(RR_{AU|e(X)=e(x)},MR_{UZ|e(X)=e(x)})$ are prespecified sensitivity analysis parameters. Then, we take the maximum of bounding factors across all propensity score values, defined as $BF_U^*=\underset{e(x)}{\max}(BF_{U|e(X)=e(x)})$, and the corresponding sensitivity analysis parameters in $BF_U^*$ are denoted as $(RR_{AU},MR_{UZ})$. Let $ACE_{AZ}^{true}$ denotes the average causal effect. When treatment effect is positive, we have
\begin{equation}\label{eq:sa_posi}
ACE_{AZ}^{true}\geq \frac{1}{2}(1+\frac{1}{BF_U^*})E(Z|A=1)-\frac{1}{2}(1+BF_U^*)E(Z|A=0).  
\end{equation}
When treatment effect is negative we have
\begin{equation}\label{eq:sa_nega}
ACE_{AZ}^{true}\leq \frac{1}{2}(1+BF_U^*)E(Z|A=1)-\frac{1}{2}(1+\frac{1}{BF_U^*})E(Z|A=0).
\end{equation}

\subsection{Interpreting the Sensitivity Analysis} 

For an unmeasured confounding with prespecified magnitude of $(RR_{AU},MR_{UZ})$, the bounding factor $BF_U^*$ can be computed by equation (\ref{eq:bf}). For a positive treatment effect, the lower bound of the treatment effect can be computed by equation (\ref{eq:sa_posi}). A positive lower bound indicates that there is still a positive treatment effect with an unmeasured confounding effect of magnitude $(RR_{AU}, MR_{UZ})$. A non-positive lower bound indicates that the positive treatment effect could be explained away by the unmeasured confounding of magnitude $(RR_{AU}, MR_{UZ})$. For a negative treatment effect, the upper bound of the treatment effect can be computed by equation (\ref{eq:sa_nega}), and similar interpretations can be made. A negative upper bound indicates that there is still a negative treatment effect with an unmeasured confounding effect of magnitude $(RR_{AU}, MR_{UZ})$. A non-negative upper bound indicates that the negative treatment effect could be explained away by the unmeasured confounding of magnitude $(RR_{AU}, MR_{UZ})$. 


\section{Real Data Example}
\label{s:real}


In this section, we apply our proposed method to the ARIC data \citep{aric1989atherosclerosis} to examine the causal effect of baseline smoking on stroke-free survival. Incident ischemic stroke events or death, the primary outcome, are identified through December 31, 2011. After excluding a small portion of subjects with missing values in the variables of interest, the total sample size used in the analysis is 14,549. The event time is defined as the follow-up time (in months) for the first incident stroke or death, whichever comes first, and a subject is censored if neither incident stroke nor death is observed during the study. 
There are 5345 events, corresponding to a 63.3$\%$ censoring rate. Given the length of follow-up, we choose 240 months as the truncation time $\tau$ for the RMST calculation. Exposure is defined as the smoking status at baseline. There are 3,832 (26.3$\%$) current smokers at baseline. Eight important baseline covariates are included in the propensity score model: race (black, white), gender (male, female), age (44-66 yrs old), BMI (14.2-65.9), diabetes (1=yes, 0=no) , HDL (10-163 mg/dL), LDL (0-504.6 mg/dL), and hypertension (1=yes, 0=no). Table \ref{tab:c1} summarizes these variables by baseline smoking status.

We first fit a logistic regression model on baseline smoking status using the eight covariates to estimate the propensity score. Then, we conduct a 1-1 optimal pair matching without replacement for all subjects which results in 3832 pairs and unmatched nonsmokers are removed from the matched sample. The covariates balance is measured by the standardized mean difference, and Figure \ref{sa_balance} shows the covariates balance of ARIC data before and after propensity score matching, which indicates our matching achieves very good covariates balance.

For comparison purposes, the analysis results of the proposed method and the IPTW Cox regression method are both included. All results are summarized in Table \ref{tab:senana-2}. Both methods show significant evidence of a harmful effect of smoking on the risk of incident ischemic stroke or death. This conclusion agrees with previous findings in the literature. The matched RMST analysis suggests an average reduction of 22.3 stroke-free survival months for baseline smokers had they not smoked at the baseline. The IPTW Cox regression measures the treatment effect on the hazard ratio scale which is not directly comparable to RMST differences. The estimated HR of 2.2 implies that smoking increases the hazard of incident ischemic stroke or death.


All the above analyses assume ignorable treatment assignment. However, for such a large observational study, unmeasured confounding is likely to be present, especially given that we are only able to control a small number of factors. Therefore, it is important to assess how the observed causal effect may change in the presence of hidden bias. A sensitivity analysis, as described in section \ref{rmst_sa}, is carried out for different possible impacts of $U$ on the exposure and the outcome. Since the observed RMST difference is negative, we use equation (\ref{eq:sa_nega}) to conduct the sensitivity analysis, where we calculate the upper bound of the true causal RMST effect for different combinations of $RR_{AU}$ and $MR_{UZ}$. The results are depicted in the contour plot in Figure \ref{sa_plot}. The solid curve in the middle of the graph represents the contour where the upper bound is zero. The darker color region to the upper-right of the zero-curve indicates a positive upper bound of the true effect, and the lighter color region to the lower-left of the zero-curve indicates a negative upper bound of the true effect. 
For moderate deviations from the ignorability assumption ($RR_{AU}<1.47$ and $MR_{UZ}<1.47$), a harmful effect still holds, as the upper bound is below the solid zero-curve. For moderate-to-large deviations, the upper bound of the treatment effect may exceed zero, indicating a possibility of a null effect. For example, at (1.5, 1.5) the upper bound of the estimated treatment effect is 2.04, which indicates that the harmful treatment effect could be totally explained away by the unmeasured confounding of magnitude $(RR_{AU}, MR_{UZ})=(1.5, 1.5)$. Overall, our sensitivity analysis indicates that the observed significant causal effect is moderately robust to hidden bias.


\section{Discussion}
\label{s:discuss}

In this paper, we adopt the RMST difference as a marginal causal effect measure for survival data, since it is collapsible and has an easy interpretation. We develop a matching-based RMST difference estimator that is asymptotically unbiased and does not rely on the PH assumption. But this does not rule out the use of hazard in causal analysis with survival data. As pointed out by \citet{doescox}, the hazard function $h(t,x,z)$ may have a valid causal interpretation, if it satisfies some additive structural constraint. 

One limitation of our work is that the proposed nonparametric estimator may not be easily extended to more complex matching designs, such as 1-k or full matching designs. This is because we need to compute the covariance to account for the correlation in matched sets. But the covariance calculation relies on the assumption of equal sample sizes in both groups \citep{murray2000variance}. Therefore the covariance formula can not be applied directly to other matching designs. One strategy to relax this limitation is to consider fitting a parametric RMST regression model after matching. This could be more advantageous if we have a good idea about the outcome model specification, as it may correct residual confounding bias not captured by matching. This adds more flexibility to post-matching inference, as it can lead to more robust or efficient semiparametric strategies by combining matching with regression models \citep{Rubin1973theuse}.  It also makes our method more attractive in practice than Conner's method as the latter solely relies on KM estimation of survival functions and cannot include regression models.

\section{Acknowledgments}

This work was partially supported by grant DMS-2015552 from National Science Foundation.The Atherosclerosis Risk in Communities study has been funded in whole or in part with Federal funds from the National Heart, Lung, and Blood Institute, National Institutes of Health, Department of Health and Human Services, under Contract nos. (HHSN268201700001I, HHSN268201700002I, HHSN268201700003I, HHSN268201700005I, HHSN268201700004I). The authors thank the staff and participants of the ARIC study for their important contributions.


\bibliographystyle{unsrtnat}
\bibliography{matching_ref}

\section{Figures and Tables }

\begin{figure}[!p]
 \centerline{\includegraphics[width=0.8\textwidth]{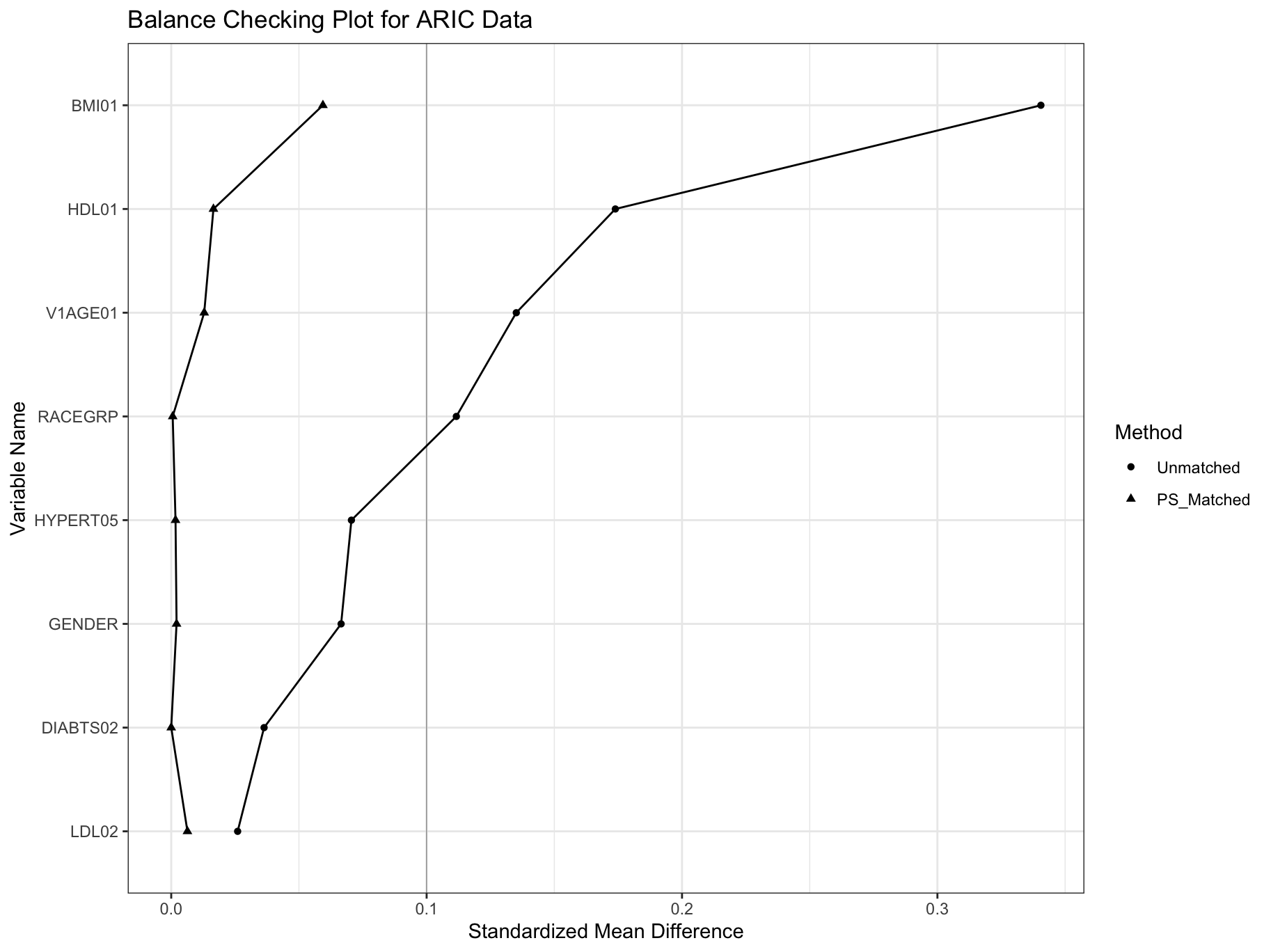}}
\caption{Covariates Balance Checking}
\label{sa_balance}
\end{figure}

\begin{figure}[!p]
 \centerline{\includegraphics[width=0.8\textwidth]{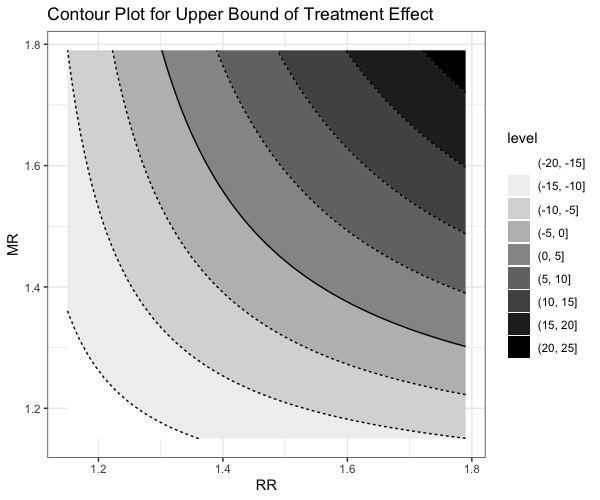}}
\caption{Contour plot of the upper bound of estimated treatment effect. The solid curve represents value 0.}
\label{sa_plot}
\end{figure}

\begin{table}[htbp]
  \centering
  \caption{Parameter Setup for Observational Studies Simulation Scenarios}
    \begin{tabular}{ccccccc}
    \toprule
    \multicolumn{2}{c}{PH Scenario Setup} & \multirow{2}[2]{*}{$\beta_A$} & \multicolumn{4}{c}{Censoring Parameter $\gamma$} \\
    ($\nu_0$,$\lambda_0$) & ($\nu_1$,$\lambda_1$) &       & 0\%   & 20\%  & 40\%  & 60\% \\
    \midrule
    (1, exp(-6)) & (1, exp(-6)) & 0     & 1.00E-08 & 0.0051 & 0.0142 & 0.0467 \\
    (1, exp(-6)) & (1, exp(-6)) & -0.4  & 1.00E-08 & 0.00462 & 0.0124 & 0.0345 \\
    (1, exp(-6)) & (1, exp(-6)) & -0.8  & 1.00E-08 & 0.00421 & 0.011 & 0.0272 \\
    (1, exp(-6)) & (1, exp(-6)) & -1.2  & 1.00E-08 & 0.00387 & 0.00992 & 0.0226 \\
    (1, exp(-6)) & (1, exp(-6)) & -2    & 1.00E-08 & 0.00322 & 0.00793 & 0.0165 \\
    \midrule
    \multicolumn{2}{c}{NP Scenario Setup} & \multirow{2}[2]{*}{$\beta_A$} & \multicolumn{4}{c}{Censoring Parameter $\gamma$} \\
    ($\nu_0$,$\lambda_0$) & ($\nu_1$,$\lambda_1$) &       & 0\%   & 20\%  & 40\%  & 60\% \\
    \midrule
    (1, exp(-6)) & (1, exp(-6)) & 0     & 1.00E-08 & 0.0049 & 0.01365 & 0.04351 \\
    (1, exp(-6)) & (1.5, 1.23E-04) & -0.4  & 1.00E-08 & 0.00346 & 0.00857 & 0.0179 \\
    (1, exp(-6)) & (1.5, 1.23E-04) & -0.8  & 1.00E-08 & 0.00323 & 0.00792 & 0.01602 \\
    (1, exp(-6)) & (1.5, 1.23E-04) & -1.2  & 1.00E-08 & 0.00306 & 0.00742 & 0.0146 \\
    (1, exp(-6)) & (1.5, 1.23E-04) & -2    & 1.00E-08 & 0.00278 & 0.00659 & 0.0126 \\
    \bottomrule
    \end{tabular}%
  \label{tab:obs_setup}%
\end{table}%

\begin{sidewaystable}
\small
\caption{Simulation Results under PH Scenario. Under zero treatment effect scenarios, bias is reported instead of percentage bias.}
\label{tab:obsph-1}
\begin{adjustbox}{scale=0.95,center}
\begin{tabular}{cccccccccccccccccc}
\hline
\multicolumn{2}{c}{\textbf{Scenario}} &
  \textbf{Bias\%} &
  \textbf{CP} &
  \textbf{SEM} &
  \textbf{SEE} &
  \textbf{Bias\%} &
  \textbf{CP} &
  \textbf{SEM} &
  \textbf{SEE} &
  \textbf{Bias\%} &
  \textbf{CP} &
  \textbf{SEM} &
  \textbf{SEE} &
  \textbf{Bias\%} &
  \textbf{CP} &
  \textbf{SEM} &
  \textbf{SEE} \\ \hline
\textbf{$\beta_A$} &
  \textbf{CR} &
  \multicolumn{4}{c}{Matched RMST (Murray)} &
  \multicolumn{4}{c}{RMST Regression} &
  \multicolumn{4}{c}{Conner's IPTW RMST} &
  \multicolumn{4}{c}{IPTW   Cox (HR)} \\ \hline
\multirow{4}{*}{0}    & 0   & 0.032   & 0.964 & 2.795 & 2.611 & -0.020   & 0.932 & 1.322 & 1.402 & 0.076   & 0.962 & 2.284 & 2.127 & 0.005    & 0.938 & 0.052 & 0.052 \\
                      & 0.2 & 0.140   & 0.954 & 2.884 & 2.686 & 0.124    & 0.938 & 1.478 & 1.542 & 0.165   & 0.956 & 2.355 & 2.218 & 0.005    & 0.924 & 0.067 & 0.070 \\
                      & 0.4 & 0.252   & 0.958 & 3.071 & 2.908 & 0.333    & 0.908 & 1.906 & 2.033 & 0.258   & 0.956 & 2.505 & 2.420 & 0.008    & 0.936 & 0.074 & 0.077 \\
                      & 0.6 & 0.288   & 0.953 & 4.079 & 4.023 & 0.877    & 0.772 & 4.450 & 6.684 & 0.344   & 0.947 & 3.369 & 3.439 & 0.013    & 0.928 & 0.085 & 0.088 \\ \hline
\multirow{4}{*}{-0.4} & 0   & 0.724\% & 0.970 & 2.795 & 2.588 & 2.746\%  & 0.944 & 1.326 & 1.369 & 1.581\% & 0.958 & 2.284 & 2.101 & -1.069\% & 0.944 & 0.053 & 0.052 \\
                      & 0.2 & 3.061\% & 0.958 & 2.875 & 2.646 & 5.384\%  & 0.938 & 1.475 & 1.487 & 3.393\% & 0.958 & 2.348 & 2.173 & -1.055\% & 0.936 & 0.069 & 0.072 \\
                      & 0.4 & 5.198\% & 0.962 & 3.032 & 2.826 & 8.447\%  & 0.934 & 1.846 & 1.906 & 5.198\% & 0.958 & 2.473 & 2.342 & -1.430\% & 0.946 & 0.076 & 0.077 \\
                      & 0.6 & 4.835\% & 0.958 & 3.647 & 3.564 & 18.611\% & 0.912 & 3.893 & 4.165 & 7.658\% & 0.958 & 2.980 & 2.908 & -2.195\% & 0.944 & 0.086 & 0.087 \\ \hline
                      
\multirow{4}{*}{-0.8} & 0   & 0.364\% & 0.970 & 2.785 & 2.570 & 3.819\%  & 0.938 & 1.329 & 1.330 & 0.794\% & 0.962 & 2.271 & 2.078 & -0.453\% & 0.944 & 0.055 & 0.053 \\
                      & 0.2 & 1.456\% & 0.962 & 2.859 & 2.657 & 4.645\%  & 0.942 & 1.475 & 1.460 & 1.703\% & 0.966 & 2.329 & 2.163 & -0.420\% & 0.938 & 0.072 & 0.074 \\
                      & 0.4 & 2.352\% & 0.966 & 2.991 & 2.721 & 6.364\%  & 0.938 & 1.805 & 1.820 & 2.359\% & 0.966 & 2.435 & 2.236 & -0.540\% & 0.940 & 0.080 & 0.080 \\
                      & 0.6 & 1.972\% & 0.960 & 3.404 & 3.297 & 4.613\%  & 0.942 & 3.236 & 3.375 & 3.756\% & 0.958 & 2.771 & 2.685 & -0.811\% & 0.942 & 0.089 & 0.089 \\ \hline
\multirow{4}{*}{-1.2} & 0   & 0.144\% & 0.966 & 2.765 & 2.550 & 4.838\%  & 0.934 & 1.333 & 1.304 & 0.481\% & 0.958 & 2.245 & 2.072 & -0.265\% & 0.948 & 0.058 & 0.057 \\
                      & 0.2 & 0.947\% & 0.962 & 2.832 & 2.621 & 5.196\%  & 0.930 & 1.476 & 1.442 & 1.145\% & 0.960 & 2.298 & 2.153 & -0.130\% & 0.934 & 0.077 & 0.078 \\
                      & 0.4 & 1.644\% & 0.966 & 2.947 & 2.640 & 5.883\%  & 0.924 & 1.778 & 1.753 & 1.700\% & 0.972 & 2.391 & 2.196 & -0.245\% & 0.940 & 0.084 & 0.084 \\
                      & 0.6 & 1.676\% & 0.968 & 3.252 & 3.045 & 4.337\%  & 0.924 & 2.849 & 2.946 & 2.473\% & 0.950 & 2.638 & 2.539 & -0.350\% & 0.942 & 0.093 & 0.094 \\ \hline
\multirow{4}{*}{-2}   & 0   & 0.077\% & 0.964 & 2.698 & 2.531 & 6.604\%  & 0.854 & 1.338 & 1.308 & 0.296\% & 0.960 & 2.160 & 2.042 & -0.098\% & 0.958 & 0.067 & 0.065 \\
                      & 0.2 & 0.407\% & 0.952 & 2.757 & 2.541 & 6.494\%  & 0.878 & 1.480 & 1.468 & 0.594\% & 0.956 & 2.206 & 2.077 & 0.069\%  & 0.950 & 0.089 & 0.087 \\
                      & 0.4 & 0.601\% & 0.966 & 2.848 & 2.615 & 6.062\%  & 0.906 & 1.748 & 1.776 & 0.748\% & 0.958 & 2.279 & 2.157 & -0.060\% & 0.954 & 0.098 & 0.095 \\
                      & 0.6 & 0.588\% & 0.970 & 3.043 & 2.837 & 4.828\%  & 0.922 & 2.474 & 2.574 & 0.917\% & 0.960 & 2.436 & 2.290 & -0.229\% & 0.944 & 0.107 & 0.105 \\ \hline
\end{tabular}
\end{adjustbox}
\end{sidewaystable}

\begin{sidewaystable}
\small
  \caption{Simulation Results under NP Scenarios. Under zero treatment effect scenarios, bias is reported instead of percentage bias.}
  \label{tab:obsnp-1}
  \begin{adjustbox}{scale=0.95,center}
\begin{tabular}{cccccccccccccccccc}
\hline
\multicolumn{2}{c}{\textbf{Scenario}} &
  \textbf{Bias\%} &
  \textbf{CP} &
  \textbf{SEM} &
  \textbf{SEE} &
  \textbf{Bias\%} &
  \textbf{CP} &
  \textbf{SEM} &
  \textbf{SEE} &
  \textbf{Bias\%} &
  \textbf{CP} &
  \textbf{SEM} &
  \textbf{SEE} &
  \textbf{Bias\%} &
  \textbf{CP} &
  \textbf{SEM} &
  \textbf{SEE} \\ \hline
\textbf{$\beta_A$} &
  \textbf{CR} &
  \multicolumn{4}{c}{Matched   RMST (Murray)} &
  \multicolumn{4}{c}{RMST Regression} &
  \multicolumn{4}{c}{Conner's IPTW RMST} &
  \multicolumn{4}{c}{IPTW   Cox (HR)} \\ \hline
\multirow{4}{*}{0}    & 0   & 0.032   & 0.964 & 2.795 & 2.611 & -0.020  & 0.932 & 1.322 & 1.402 & 0.076   & 0.962 & 2.284 & 2.127 & 0.005     & 0.938 & 0.052 & 0.052 \\
                      & 0.2 & 0.140   & 0.954 & 2.884 & 2.686 & 0.124   & 0.938 & 1.478 & 1.542 & 0.165   & 0.956 & 2.355 & 2.218 & 0.005     & 0.924 & 0.067 & 0.070 \\
                      & 0.4 & 0.252   & 0.958 & 3.071 & 2.908 & 0.333   & 0.908 & 1.906 & 2.033 & 0.258   & 0.956 & 2.505 & 2.420 & 0.008     & 0.936 & 0.074 & 0.077 \\
                      & 0.6 & 0.288   & 0.953 & 4.079 & 4.023 & 0.877   & 0.772 & 4.450 & 6.684 & 0.344   & 0.947 & 3.369 & 3.439 & 0.013     & 0.928 & 0.085 & 0.088 \\ \hline
\multirow{4}{*}{-0.4} & 0   & 0.196\% & 0.964 & 2.674 & 2.473 & 6.560\% & 0.888 & 1.244 & 1.215 & 0.428\% & 0.964 & 2.131 & 1.966 & 102.637\% & 0.000 & 0.056 & 0.061 \\
& 0.2 & 0.925\% & 0.962 & 2.746 & 2.507 & 6.798\% & 0.894 & 1.384 & 1.337 & 1.063\% & 0.964 & 2.193 & 2.019 & 298.860\% & 0.000 & 0.079 & 0.081 \\
                      & 0.4 & 1.728\% & 0.966 & 2.863 & 2.577 & 7.145\% & 0.898 & 1.653 & 1.636 & 1.781\% & 0.964 & 2.293 & 2.113 & 363.911\% & 0.000 & 0.091 & 0.095 \\
                      & 0.6 & 2.152\% & 0.972 & 3.119 & 2.897 & 6.678\% & 0.904 & 2.400 & 2.552 & 2.615\% & 0.952 & 2.516 & 2.376 & 422.119\% & 0.000 & 0.105 & 0.109 \\ \hline
\multirow{4}{*}{-0.8} & 0   & 0.158\% & 0.966 & 2.644 & 2.461 & 7.174\% & 0.820 & 1.253 & 1.212 & 0.345\% & 0.958 & 2.091 & 1.954 & 33.369\%  & 0.004 & 0.057 & 0.062 \\
                      & 0.2 & 0.608\% & 0.968 & 2.711 & 2.466 & 7.158\% & 0.864 & 1.392 & 1.349 & 0.747\% & 0.966 & 2.148 & 1.991 & 139.464\% & 0.000 & 0.083 & 0.086 \\
                      & 0.4 & 1.250\% & 0.968 & 2.816 & 2.538 & 7.139\% & 0.888 & 1.650 & 1.656 & 1.345\% & 0.962 & 2.239 & 2.079 & 172.568\% & 0.000 & 0.097 & 0.099 \\
                      & 0.6 & 1.713\% & 0.972 & 3.030 & 2.830 & 6.392\% & 0.906 & 2.302 & 2.383 & 1.941\% & 0.958 & 2.424 & 2.286 & 201.225\% & 0.000 & 0.111 & 0.116 \\ \hline
\multirow{4}{*}{-1.2} & 0   & 0.079\% & 0.966 & 2.605 & 2.434 & 7.806\% & 0.724 & 1.261 & 1.224 & 0.263\% & 0.960 & 2.042 & 1.925 & 10.201\%  & 0.446 & 0.059 & 0.063 \\
                      & 0.2 & 0.422\% & 0.962 & 2.669 & 2.440 & 7.695\% & 0.794 & 1.400 & 1.382 & 0.597\% & 0.962 & 2.095 & 1.966 & 86.905\%  & 0.000 & 0.088 & 0.090 \\
                      & 0.4 & 0.911\% & 0.970 & 2.764 & 2.511 & 7.386\% & 0.852 & 1.650 & 1.676 & 0.984\% & 0.950 & 2.177 & 2.022 & 109.380\% & 0.000 & 0.103 & 0.107 \\
                      & 0.6 & 1.353\% & 0.972 & 2.946 & 2.776 & 6.397\% & 0.904 & 2.230 & 2.284 & 1.438\% & 0.952 & 2.335 & 2.234 & 128.294\% & 0.000 & 0.118 & 0.123 \\ \hline
\multirow{4}{*}{-2}   & 0   & 0.053\% & 0.968 & 2.514 & 2.362 & 9.169\% & 0.512 & 1.279 & 1.238 & 0.204\% & 0.958 & 1.923 & 1.827 & -8.274\%  & 0.274 & 0.063 & 0.067 \\
                      & 0.2 & 0.221\% & 0.964 & 2.569 & 2.365 & 8.835\% & 0.622 & 1.419 & 1.419 & 0.372\% & 0.954 & 1.969 & 1.860 & 45.487\%  & 0.000 & 0.101 & 0.101 \\
                      & 0.4 & 0.533\% & 0.956 & 2.648 & 2.452 & 8.418\% & 0.730 & 1.653 & 1.689 & 0.607\% & 0.954 & 2.036 & 1.942 & 59.417\%  & 0.000 & 0.117 & 0.122 \\
                      & 0.6 & 0.782\% & 0.956 & 2.791 & 2.623 & 7.366\% & 0.866 & 2.149 & 2.098 & 0.831\% & 0.948 & 2.157 & 2.105 & 70.605\%  & 0.000 & 0.133 & 0.136 \\ \hline
\end{tabular}

\end{adjustbox}
\end{sidewaystable}

\begin{table}[!p]
\caption{Summary Statistics of Covariates by Baseline Smoking Status in ARIC Study}
\label{tab:c1}
\begin{tabular}{lll}
\hline
                         & Non-current smoker (10,717) & Current smoker (3,832) \\ \hline
Race, n (\%) of white    & 8161 (76.2\%)               & 2730 (71.2\%)          \\
Gender, n (\%) of female & 5957 (55.6\%)               & 2003 (52.3\%)          \\
Age, mean (SD)           & 54.4 (5.8)                  & 53.7 (5.7)             \\
BMI, mean (SD)           & 28.1 (5.4)                  & 26.3 (5.0)             \\
Diabetes, n (\%)         & 1044 (9.7\%)                & 333 (8.7\%)            \\
HDL (mmol/L), mean (SD)  & 52.6 (16.8)                 & 49.6 (17.3)            \\
LDL (mmol/L), mean (SD)  & 137.6 (38.9)                & 138.6 (40.4)           \\
Hypertension, n (\%)     & 3783 (35.3\%)               & 1225 (32.0\%)          \\ \hline
\end{tabular}
\end{table}

\begin{table}[!p]
\caption{ARIC Data Analysis Results}
\label{tab:senana-2}
\begin{tabular}{lcccc}
\hline
\multicolumn{1}{c}{} & Estimate & SE & 95\% CI Lower Bound & 95\% CI Upper Bound \\
\hline
Matched RMST        & -22.266 & 1.380 & -24.972 & -19.561 \\
IPTW Cox  (HR)      & 2.173   & 0.066 & 2.048   & 2.306\\
\hline
\end{tabular}
\end{table}

\clearpage
\pagenumbering{arabic}
\renewcommand*{\thepage}{Appendix-\arabic{page}}


\section*{Web Appendix A: Proofs of Theoretical Results in Section 2 of the Main Text}
\label{s:w_a}

\setcounter{proposition}{0}


We will prove the propositions and related lemmas in Section 2 of the main text.
\begin{proposition}
Given assumptions \ref{sutva}-\ref{strong_ignore}, we have $(T^0,T^1)\indep A|e(X)$, which further implies $(Z^0,Z^1)\indep A|e(X)$.
\end{proposition}
\begin{proof} 
It is equivalent to show 
$P\{A=1|T^1,T^0,e(X)\}=P\{A=1|e(X)\}$.
By Theorem 2 in \citet{cps}, we have $P(A=1|e(X))=E\{e(X)|e(X)\}=e(X)$, then it is equivalent to show 
$P\{A=1|T^1,T^0,e(X)\}=e(X)$. We have
\begin{align*}
P\{A=1|T^1,T^0,e(X)\}&=E\{P(A=1|T^1,T^0,X)|T^1,T^0,e(X)\}\\
&=E\{P(A=1|X)|T^1,T^0,e(X)\}\text{(by strongly ignorability)}\\
&=E\{e(X)|T^1,T^0,e(X)\}=e(X)=P\{A=1|e(X)\}.
\end{align*}
Thus, we have $(T^0,T^1)\indep A|e(X)$ for $0<\text{pr}(A=1|e(X))<1$. Since $Z^A=\min (T^A, \tau)$ and $\tau$ is a fixed constant, the above conditional independence also implies
$(Z^0,Z^1)\indep A|e(X)$. 
\end{proof}
\begin{lemma}
\label{marginal_ignorable}
Given assumptions \ref{sutva}-\ref{indep_censor}, $(Y^1, Y^0)\indep A$ holds marginally in the matched sample under the propensity score matching design.
\end{lemma}%
\begin{proof}
By assumption~\ref{strong_ignore}, we have $(Y^1, Y^0)\indep A|e(X)$ where $0<P(A=1|e(X))<1$. Let $M$ denotes the matching structure, and $\epsilon_M$ denotes the set of propensity scores in the matched sample.
Then, we have the following equation by matching on propensity score $e(X)$ with a constant treatment to control allocation ratio $1: k$ ($k=1$ for pair matching),
 $$P(A=1|e(X))=\frac{1}{k+1},\ \text{for all $e(X)\in\epsilon_M$}.$$
Thus, $e(X)\indep A$ holds in the matched sample, i.e. $f_M(e(X)|A)=f_M(e(X))$. 
Consider the joint density of $Y^1$ and $Y^0$ conditional on $A$ in the matched sample, which is denoted as $f_M(Y^1, Y^0|A)$, we have
\begin{align*}
&f_M(Y^1, Y^0|A)=\int_{\epsilon_M}f(Y^1, Y^0|A, e(X))f_M(e(X)|A)de(X)\\
=&\int_{\epsilon_M}f(Y^1, Y^0|A, e(X))f_M(e(X))de(X)\ \text{[matched by constant allocation ratio]}\\
=&\int_{\epsilon_M}f(Y^1,Y^0|e(X))f_M(e(X))de(X)\ \text{[by assumption~\ref{strong_ignore}]}\\
=&f_M(Y^1, Y^0).
\end{align*}
Since $f_M(Y^1, Y^0|A)=f_M(Y^1, Y^0)$ implies $(Y^1, Y^0)\indep A$ in the matched sample, $(Y^1, Y^0)\indep A$ holds marginally in the matched sample.
\end{proof}
\begin{lemma}
\label{rmst_lemma}
Let $\hat{S}_{e(X),A}(t)$ denotes the KM survival function estimator given propensity score $e(X)$ and treatment indicator $A$.
For a fixed truncation time $\tau$, 
$$\underset{n\rightarrow\infty}{\lim}\int^\tau_0E_T[\hat{S}_{e(X),A}(t)-S_{e(X),A}(t)]dt=0,$$
\end{lemma}
\begin{proof}
Define $\tilde{T}=\min (T,C)$ and $\pi_{e(X),A}(t)=P(\tilde{T}\geq t)\in (0,1)$, then $[1-\pi_{e(X),A}(t)]^n$ is a nonnegative function that increases as $t$ increases. By Lemma 3.2.1 in \citet{fleming2011counting}, we know:
\begin{align*}
&\int^\tau_0E_T[\hat{S}_{e(X),A}(t)-S_{e(X),A}(t)]dt
\leq\int^\tau_0[1-S_{e(X),A}(t)][1-\pi_{e(X),A}(t)]^ndt\\
&\leq\int^\tau_0[1-\pi_{e(X),A}(t)]^ndt \leq\tau[1-\pi_{e(X),A}(\tau)]^n.
\end{align*}
Since $\tau>0$ is a fixed constant and $1-\pi_{e(X),A}(\tau)\in(0,1)$, we have
$$\underset{n\rightarrow\infty}{\lim}\int^\tau_0E_T[\hat{S}_{e(X),A}(t)-S_{e(X),A}(t)]dt\leq\underset{n\rightarrow\infty}{\lim}\tau[1-\pi_{e(X),A}(\tau)]^n=0.$$
Therefore, we have $\underset{n\rightarrow\infty}{\lim}\int^\tau_0E_T[\hat{S}_{e(X),A}(t)-S_{e(X),A}(t)]dt=0$. 
\end{proof}
\begin{proposition}
Given assumptions \ref{sutva}-\ref{tau_assump}, the RMST estimator based on KM method given propensity score $e(X)$ and treatment group $A$, denoted as $\hat{\mu}_{e(X),A}$, is an asymptotically unbiased estimator for $\mu_{e(X),A}$ given $\tau<t_{\max}$.
\end{proposition}
\begin{proof}%
First, we will show that $\hat{S}_{e(X),A}(t)$ is asymptotically unbiased for any time $T< t_{\max}$. Let $t_i$'s be i.i.d event times ranking from small to large, and $Y_i$ is the number of people at risk at event time $t_i$. Let $d_i$ be the number of event at event time $t_i$, then we have the definition below.
\begin{equation*}
\hat{S}_{e(X),A}(t)
=\left\{\begin{array}{l}
1, \text{if}\ t\leq t_1\ \text{given}\ e(X)\text{ and } A\\
\underset{t_i\leq t}{\prod}\frac{Y_i-d_i}{Y_i}, \text{if}\ t_1\leq t \ \text{given}\ e(X)\text{ and }A\\
\end{array}\right.
\end{equation*}
Let $\hat{\Lambda}_{e(X),A}(u)=\underset{t_i<t}{\sum}\frac{d_i}{Y_i}$ be the Nelson-Aalen estimator for the cumulative hazard function $\Lambda_{e(X),A}(u)$ given $e(X)$ and $A$.
According to Theorem 3.2.3 in \citet{fleming2011counting}, we have the following equation if $S_{e(X),A}(t)>0$:
\begin{align*}
    &\frac{\hat{S}_{e(X),A}(t)}{S_{e(X),A}(t)}-1=-\int^t_0\frac{\hat{S}_{e(X),A}(u^-)}{S_{e(X),A}(u)}d\{\hat{\Lambda}_{e(X),A}(u)-\Lambda_{e(X),A}(u)\},\\
    &E[\hat{S}_{e(X),A}(t)-S_{e(X),A}(t)]=E[I_{\{T<t\}}\frac{\hat{S}_{e(X),A}(T)\{S_{e(X),A}(T)-S_{e(X),A}(t)\}}{S_{e(X),A}(T)}].
\end{align*}

Based on Lemma 3.2.1 in \citet{fleming2011counting}, the bias $E[\hat{S}_{e(X),A}(t)-S_{e(X),A}(t)]$ will converge to zero as sample size $n\rightarrow \infty$. Thus, $\hat{S}_{e(X),A}(t)$ is asymptotically unbiased given $t<t_{\max}$. Similarly, $\hat{S}_{e(X),A=0}(t)$ is also asymptotically unbiased given $t<t_{\max}$.

Second, we will show $\hat{\mu}_{e(X),A}$ is an asymptotically unbiased estimator given $\tau<t_{\max}$.
Since $E_{T}(\hat{\mu}_{e(X),A})=E_{T}[\int^\tau_0\hat{S}_{e(X),A}(t)dt]$ and 
$\hat{S}_{e(X),A}(t)$ is a positive bounded function between 0 and 1 when $t\in[0,\tau]$, then we have
$$E_{T}[\int^\tau_0|\hat{S}_{e(X),A}(t)|dt]=E_{T}[\int^\tau_0\hat{S}_{e(X),A}(t)dt]\leq\tau<\infty.$$
By Fubini's Theorem, 
$$E_{T}[\int^\tau_0\hat{S}_{e(X),A}(t)dt]=\int^\tau_0E_{T}[\hat{S}_{e(X),A}(t)]dt.$$

By propositions~\ref{RMST_unbias} and lemma~\ref{rmst_lemma}, we have the following for fixed truncated time $\tau$.
\begin{align*}
&\underset{n\rightarrow\infty}{\lim}E_{T}(\hat{\mu}_{e(X),A})-\mu_{e(X),A}\\
=&\underset{n\rightarrow\infty}{\lim}E_{T}[\int^\tau_0\hat{S}_{e(X),A}(t)dt]-\int^\tau_0S_{e(X),A}(t)dt\\
=&\underset{n\rightarrow\infty}{\lim}\int^\tau_0E_{T}[\hat{S}_{e(X),A}(t)]dt-\int^\tau_0S_{e(X),A}(t)dt\ \text{(by Fubini's Theorem)}\\
=&\underset{n\rightarrow\infty}{\lim}\int^\tau_0E_{T}\{\hat{S}_{e(X),A}(t)-S_{e(X),A}(t)\}dt
\leq\underset{n\rightarrow\infty}{\lim}\tau[1-\pi_{e(X),A}(\tau)]^n=0.
\end{align*}
Therefore, $\hat{\mu}_{e(X),A}$ is an asymptotically unbiased estimator for $\mu_{e(X),A}$ when $\tau<t_{\max}$.
\end{proof}
\begin{lemma}
\label{prop:asyunbias_lemma}
For a fixed truncation time point $\tau<t_{\max}$, 
$$\underset{n\rightarrow\infty}{\lim}\int^\tau_0E_T[\hat{S}^AA(t)-S^A(t)]dt=0,$$
\end{lemma}
\begin{proof}
Let $\tilde{T}=\min (T,C)$ and $\pi(t)=P(\tilde{T}\geq t)\in (0,1)$, then $[1-\pi(t)]^n$ is a nonnegative function which increases as $t$ increases. Let $\pi_A(t)$ denotes the function $\pi(t)$ for treatment indicator A. By Lemma 3.2.1 in \citet{fleming2011counting}, we have
$$\int^\tau_0E_T[\hat{S}^A(t)-S^A(t)]dt\leq \int^\tau_0 [1-S^A(t)][1-\pi_A(t)]^ndt\leq\int^\tau_0[1-\pi_A(t)]^ndt\leq \tau[1-\pi_A(t)]^n.$$
Since $\tau>0$ is a fixed constant and $1-\pi_A(\tau)\in(0,1)$, we have
$$\underset{n\rightarrow\infty}{\lim}\int^\tau_0E_T[\hat{S}^A(t)-S^A(t)]dt\leq\underset{n\rightarrow\infty}{\lim}\tau[1-\pi_A(t)]^n=0.$$
\end{proof}
\begin{proposition}
Given assumptions \ref{sutva}-\ref{tau_assump},  $\hat{\Delta}_{ATT}=\int^\tau_0[\hat{S}^1(t)-\hat{S}^0(t)]dt$ is asymptotically unbiased.
\end{proposition}
\begin{proof}
Since the estimated survival function $\hat{S}(t)\in(0,1)$ and $|\int^\tau_0\hat{S}(t)dt|\in (0,\tau)$, we also satisfy the following conditions to use Fubini's theorem:
\begin{enumerate}
\item $E_{e(X)}\{E_T|\int^\tau_0\hat{S}_{e(X),A}(t)dt|\}\leq E_{e(X)}\{E_T(\tau)\}=\tau<\infty$ 
\item $E_{e(X)}[\int^\tau_0|\hat{S}_{e(X),A}(t)|dt]\leq \tau<\infty$ 
\item $E_T\{\int^\tau_0\hat{S}^A(t)dt\}\leq\tau<\infty$ 
\end{enumerate}%
Thus, we can apply Fubini's theorem three time to interchange the expectation of $e(X)$ as below: 
\begin{align*}
&E_{e(X)}\{E_T[\int^\tau_0\hat{S}_{e(X),A=1}(t)dt]-E_T[\int^\tau_0\hat{S}_{e(X),A=0}(t)dt]\}\\ 
=&E_T\{E_{e(X)}[\int^\tau_0\hat{S}_{e(X),A=1}(t)dt]-E_{e(X)}[\int^\tau_0\hat{S}_{e(X),A=0}(t)dt]\}\\
=&E_T\{\int^\tau_0E_{e(X)}[\hat{S}_{e(X),A=1}(t)]dt-\int^\tau_0E_{e(X)}[\hat{S}_{e(X),A=0}(t)]dt\}\\
=&E_T\{\int^\tau_0\hat{S}^1(t)dt-\int^\tau_0\hat{S}^0(t)dt\}\\
=&\int^\tau_0E_T[\hat{S}^1(t)]dt-\int^\tau_0E_T[\hat{S}^0(t)]dt.
\end{align*}
By lemma~\ref{prop:asyunbias_lemma}, we have
\begin{align*}
&\underset{n\rightarrow\infty}{\lim}E_{e(X)}\{E_T[\int^\tau_0\hat{S}_{e(X),A=1}(t)dt]-E_T[\int^\tau_0\hat{S}_{e(X),A=0}(t)dt]\}\\ 
=&\underset{n\rightarrow\infty}{\lim}\int^\tau_0E_T[\hat{S}^1(t)]dt-\underset{n\rightarrow\infty}{\lim}\int^\tau_0E_T[\hat{S}^0(t)]dt\\
=&\int^\tau_0E_T[S^1(t)]dt-\int^\tau_0E_T[S^0(t)]dt
=\mu_1-\mu_0.
\end{align*}
Therefore, our proposed propensity score matched RMST estimator is asymptotically unbiased when truncation time point $\tau<t_{\max}$.
\end{proof}

\newpage
\section*{Web Appendix B: Proofs of Theoretical Results for Section 4 of the Main Text}
\label{s:w_b}
\subsection*{Web Appendix B.1: Proofs of Propositions about Conditional Effect}

We define the expectations of the RMST outcome $Z=\min (T,\tau)$. The following propositions are proved \textbf{within each propensity score value} $e(X)=e(x)$. 
\begin{proposition}
\label{sa_binary}
For binary unmeasured confounder $U=0,1$, we have $RR_{AU|e(X)=e(x)}\geq 1$ and $MR_{UZ|e(X)=e(x)}\geq 1$.
\end{proposition}
\begin{proof}
By definition, we have $MR_{UZ|A=1,e(X)=e(x)}\geq 1$ and $MR_{UZ|A=0,e(X)=e(x)}\geq1$, then 
$MR_{UZ|e(X)=e(x)}=\max (MR_{UZ|A=1,e(X)=e(x)},MR_{UZ|A=0,e(X)=e(x)})\geq 1.$
Assume $RR_{AU|e(x)}=\underset{k=0,1}{\max}RR_{AU,k|e(x)}<1$, then it implies that
\begin{align*}
    P(U=0|A=1,e(X)=e(x))<P(U=0|A=0,e(X)=e(x)),\\
    P(U=1|A=1,e(X)=e(x))<P(U=1|A=0,e(X)=e(x)).
\end{align*}
This further implies that $1=P(U=0|A=1,e(X)=e(x))+P(U=1|A=1,e(X)=e(x))<P(U=0|A=0,e(X)=e(x))+P(U=1|A=0,e(X)=e(x))=1$, which is not true. Thus, we have proved by contradiction that $RR_{AU|e(X)=e(x)}\geq 1$.
\end{proof}
\begin{proposition}
\label{sa_prop1}
$$CMR_{AZ^+}=\frac{MR_{AZ}}{MR_{AZ^+}^{true}}\leq BF_U,\ CMR_{AZ^-}=\frac{MR_{AZ}}{MR_{AZ^-}^{true}}\leq BF_U,\ CMR_{AZ}=\frac{MR_{AZ}}{MR_{AZ}^{true}}\leq BF_U,$$
\end{proposition}
\begin{proof}
First, let $f=P(A=1)$, then we have
\begin{align} 
MR_{AZ}^{true}=&\frac{\int r_1(u)F(du)}{\int r_0(u)F(du)}
=\frac{f\int r_1(u) F_1(du)+(1-f)\int r_1(u)F_0(du)}{f\int r_0(u) F_1(du)+(1-f)\int r_0(u)F_0(du)}\\
=&\frac{f\int r_0(u)F_1(du)}{f\int r_0(u)F_1(du)+(1-f)\int r_0(u)F_0(du)}\times\frac{\int r_1(u)F_1(du)}{\int r_0(u)F_1(du)}\\
&+\frac{(1-f)\int r_0(u)F_0(du)}{f\int r_0(u)F_1(du)+(1-f)\int r_0(u)F_0(du)}\times \frac{\int r_1(u)F_0(du)}{\int r_0(u)F_0(du)}.
\end{align}

Let $w=\frac{f\int r_0(u)F_1(du)}{f\int r_0(u)F_1(du)+(1-f)\int r_0(u)F_0(du)}\in [0,1]$, then we have 
\begin{align*}
& MR^{true}_{AZ}=wMR^{true}_{AZ^+}+(1-w)MR^{true}_{AZ^-};
& \frac{1}{CMR_{AZ}}=\frac{w}{CMR_{AZ^+}}+\frac{1-w}{CMR_{AZ^-}}
\end{align*}

Second, we have 
$$CMR_{AZ^+}=\frac{MR^{obs}_{AZ}}{MR_{AZ^+}^{true}}=\frac{\int r_1(u)F(du)}{\int r_0(u)F(du)}/\frac{\int r_1(u)F_1(du)}{\int r_0(u)F_1(du)}=\frac{w_1\underset{u}{\max}r_0(u)+(1-w_1)\underset{u}{\min}r_0(u)}{w_0\underset{u}{\max}r_0(u)+(1-w_0)\underset{u}{\min}r_0(u)}$$
where $w_1=\frac{\int[r_0(u)-\underset{u}{\min}r_0(u)]F_1(du)}{\underset{u}{\max}r_0(u)-\underset{u}{\min}r_0(u)}$ and $w_0=\frac{\int[r_0(u)-\underset{u}{\min}r_0(u)]F_0(du)}{\underset{u}{\max}r_0(u)-\underset{u}{\min}r_0(u)}$. 

Define $\Gamma=\frac{w_1}{w_0}$ then 
\begin{align}
\Gamma&=\frac{w_1}{w_0}=\frac{\int[r_0(u)-\underset{u}{\min}r_0(u)]F_1(du)}{\int[r_0(u)-\underset{u}{\min}r_0(u)]F_0(du)}
=\frac{\int[r_0(u)-\underset{u}{\min}r_0(u)]RR_{AU}(u)F_0(du)}{\int[r_0(u)-\underset{u}{\min}r_0(u)]F_0(du)}\\
&\leq \frac{\underset{u}{\max}RR_{AU}(u)\int[r_0(u)-\underset{u}{\min}r_0(u)]F_0(du)}{\int[r_0(u)-\underset{u}{\min}r_0(u)]F_0(du)}
=RR_{AU}.
\end{align}

Write $w_0=\frac{w_1}{\Gamma}$, then
$$CMR_{AZ}^+=\frac{\underset{u}{[\max}r_0(u)-\underset{u}{\min}r_0(u)]w_1+\underset{u}{\min}r_0(u)}{[\underset{u}{\max}r_0(u)-\underset{u}{\min}r_0(u)]w1/\Gamma+\underset{u}{\min}r_0(u)}.$$

If $\Gamma>1$, $CMR_{AZ}^+$ is increasing in $w_1$ according to Lemma A.1 in the eAppendix of \citet{ding2016sensitivity}, then the maximum attains at $w_1=1$, and we have
$$CMR_{AZ}^+\leq\frac{\Gamma\times MR_{UZ|A=0}}{\Gamma+MR_{UZ|A=0}-1}\leq\frac{RR_{AU}\times MR_{UZ|A=0}}{RR_{AU}+MR_{UZ|A=0}-1}.$$

If $\Gamma\leq1$, $CMR_{AZ}^+$ is non-increasing in $w_1$ according to Lemma A.1 in the eAppendix of \citet{ding2016sensitivity}, then the maximum attains at $w_1=0$, and we have
$$CMR_{AZ}^+\leq 1\leq\frac{RR_{AU}\times MR_{UZ|A=0}}{RR_{AU}+MR_{UZ|A=0}-1}.$$
Similarly, by $\frac{1}{CMR_{AZ}}=\frac{w}{CMR_{AZ^+}}+\frac{1-w}{CMR_{AZ^-}}$, we have 
$$\frac{1}{CMR_{AZ}}\geq\frac{1}{BF_U}\ ,  CMR_{AZ}\leq BF_U.$$ 
\end{proof}

To study the average causal effect of the exposure on the difference scale, we need the following definitions:
\begin{itemize}
\item Define $m_0=E(Z|A=0)$ and $m_1=E(Z|A=1)$ , then the observed mean difference of exposure on the outcome is $m_1-m_0$.
\item  The average causal effect of the exposure on the outcome for exposed is
\begin{align*} 
ACE_{AZ^+}^{true}&=\int E(Z|A=1, U=u)F_1(du)-\int E(Z|A=0, U=u)F_1(du)\\
&=m_1-\int r_0(u)F_1(du).
\end{align*}
\item The average causal effect of the exposure on the outcome for unexposed is
\begin{align*} 
ACE_{AZ^-}^{true}&=\int E(Z|A=1, U=u)F_0(du)-\int E(Z|A=0, U=u)F_0(du)\\
&=\int r_1(u)F_0(du)-m_0.
\end{align*}
\item The average causal effect of the exposure on the outcome for whole population is
\begin{align*}
ACE_{AZ}^{true}&=\int E(Z|A=1, U=u)F(du)-\int E(Z|A=0, U=u)F(du)\\
&=fACE_{AZ^+}^{true}+(1-f)ACE_{AZ^-}^{true}.
\end{align*}

\end{itemize}
\begin{proposition}
\label{sa_prop2}
For nonnegative outcomes and $ACE_{AZ}^{obs}\geq0$, the lower bounds for the average causal effects are 
\begin{align*}
& ACE_{AZ^+}^{true}\geq m_1-m_0\times BF_U; ACE_{AZ^-}^{true}\geq m_1/BF_U-m_0;\\
& ACE_{AZ}^{true}\geq (m_1-m_0\times BF_U)[f+(1-f)/BF_U]=(\frac{m_1}{BF_U}-m_0)[f\times BF_U+(1-f)].
\end{align*}
\end{proposition}
\begin{proof}
From the data, we can identify
$$m_1=\int E(Z|A=1, U=u)F_1(du)=\int r_1(u)F_1(du)=E(Z|A=1);$$
$$m_0=\int E(Z|A=0, U=u)F_0(du)=\int r_0(u)F_0(du)=E(Z|A=0).$$
The counterfactual probabilities are not identifiable:
$$E(Z(1)=1|A=0)=\int E(Z=1|A=1, U=u)F_0(du)=\int r_1(u)F_0(du);$$
$$E(Z(0)=1|A=1)=\int E(Z=1|A=0, U=u)F_1(du)=\int r_0(u)F_1(du).$$

First, by proposition~\ref{sa_prop1} we have
\begin{align*}
\frac{m_1}{E(Z(1)=1|A=0)}&=\frac{\int r_1(u)F_1(du)}{\int r_1(u)F_0(du)}
=\frac{\int r_1(u)F_1(du)}{\int r_0(u)F_0(du)}/\frac{\int r_1(u)F_0(du)}{\int r_0(u)F_0(du)}\\
&=\frac{MR_{AZ}}{MR_{AZ^-}^{true}}=CMR_{AZ^-}\leq BF_U.
\end{align*}
Thus, we have $E(Z(1=1)|A=0)\geq \frac{m_1}{BF_U}$.

Second, by proposition~\ref{sa_prop1} again we have 
$$\frac{E(Z(0)=1|A=1)}{m_0}=\frac{\int r_0(u)F_1(du)}{\int r_0(u)F_0(du)}=CMR_{AZ^+}\leq BF_U.$$
Thus, we have $E(Z(0)=1|A=1)\leq m_0 BF_U.$

By definition of ACE and the inequalities derived above, we have
\begin{align*}
ACE_{AZ^+}^{true}&=m_1-\int r_0(u)F_1(du)\geq m_1-m_0\times BF_U ;\\
ACE_{AZ^-}^{true}&=\int r_1(u)F_0(du)-m_)\geq m_1/BF_U-m_0 ;\\
ACE_{AZ}^{true}&=f\cdot ACE_{AZ^+}^{true}+(1-f)ACE_{AZ^-}^{true}\\
&\geq f(m_1-m_0BF_U)+(1-f)(\frac{m_1}{BF_U}-m_0)\\
&=(m_1-m_0\times BF_U)[f+(1-f)/BF_U]\\
&=(\frac{m_1}{BF_U}-m_0)[f\times BF_U+(1-f)].
\end{align*}
\end{proof}
\begin{proposition}
\label{sa_prop3}
For nonnegative outcomes with $ACE_{AZ}^{obs}<0$, we have
\begin{align*}
& ACE_{AZ^+}^{true}\leq m_1 BF_U-m_0 ;
ACE_{AZ^-}^{true}\leq m_1-\frac{m_0}{BF_U} ;\\
& ACE_{AZ}^{true}\leq (m_1 BF_U-m_0)(f+\frac{1-f}{BF_U})=(m_1-\frac{m_0}{BF_U})(f BF_U+1-f).
\end{align*}
\end{proposition}
\begin{proof}
Define $\bar{A}=1-A$. By applying proposition~\ref{sa_prop2} we have
\begin{align*}
ACE_{\bar{A}Z^+}^{true}&\geq E(Z|\bar{A}=1)-E(Z|\bar{A}=0)\times BF_U ;\\
ACE_{\bar{A}Z^-}^{true}&\geq E(Z|\bar{A}=1)/BF_U-E(Z|\bar{A}=0) ;\\
ACE_{\bar{A}Z}^{true}&\geq (E(Z|\bar{A}=1)-E(Z|\bar{A}=0)\times BF_U)[f+(1-f)/BF_U]\\
&=(\frac{E(Z|\bar{A}=1)}{BF_U}-E(Z|\bar{A}=0))[f\times BF_U+(1-f)].
\end{align*}
Because $ACE_{\bar{A}Z^+}^{true}=-ACE_{AZ^+}^{true}$, $ACE_{\bar{A}Z^-}^{true}=-ACE_{AZ^-}^{true}$ and $ACE_{\bar{A}Z}^{true}=-ACE_{AZ}^{true}$, and we also have $E(Z|\bar{A}=0)=E(Z|A=1)=m_1$ and $E(Z|\bar{A}=1)=E(Z|A=0)=m_0$. Then we have
\begin{align*}
& ACE_{AZ^+}^{true}\leq m_1 BF_U-m_0 ; ACE_{AZ^-}^{true}\leq m_1-\frac{m_0}{BF_U} ;\\
& ACE_{AZ}^{true}\leq (m_1 BF_U-m_0)(f+\frac{1-f}{BF_U})=(m_1-\frac{m_0}{BF_U})(f BF_U+1-f).
\end{align*}
\end{proof}
\subsection*{Web Appendix B.2: Proofs of Propositions about the Marginal Effect}

To make the bounding factor hold for all propensity score values, we consider the maximum value of $BF_U$ across all values of propensity score $e(X)$, which is defined as $BF_U^*=\underset{e(x)}{\max}(BF_{U|e(X)=e(x)})$. 
\begin{proposition}
\label{sa_prop4}
For nonnegative outcomes and $ACE_{AZ}^{obs}\geq0$, we have
\begin{align*}
& ACE_{AZ^+}^{true}\geq m_1-m_0\times BF^*_U ;
ACE_{AZ^-}^{true}\geq m_1/BF^*_U-m_0 ;\\
& ACE_{AZ}^{true}\geq (m_1-m_0\times BF^*_U)[f+(1-f)/BF^*_U]=(\frac{m_1}{BF^*_U}-m_0)[f\times BF^*_U+(1-f)].
\end{align*}
For nonnegative outcomes and $ACE_{AZ}^{obs}<0$, we have
\begin{align*}
& ACE_{AZ^+}^{true}\leq m_1 BF^*_U-m_0 ;
ACE_{AZ^-}^{true}\leq m_1-\frac{m_0}{BF^*_U} ;\\
& ACE_{AZ}^{true}\leq (m_1 BF^*_U-m_0)(f+\frac{1-f}{BF^*_U})=(m_1-\frac{m_0}{BF^*_U})(f BF^*_U+1-f).
\end{align*}
\end{proposition}
\begin{proof}
We will start from showing the results for nonnegative outcomes and $ACE_{AZ}^{obs}\geq0$.
First, we have
\begin{align*}
\frac{m_1}{E(Z(1)=1|A=0)}&=\frac{\int r_1(u)F_1(du)}{\int r_1(u)F_0(du)}
=\frac{\int r_1(u)F_1(du)}{\int r_0(u)F_0(du)}/\frac{\int r_1(u)F_0(du)}{\int r_0(u)F_0(du)}\\
&=\frac{MR_{AZ}}{MR_{AZ^-}^{true}}=CMR_{AZ^-}\leq BF_U \leq BF^*_U.
\end{align*}
Thus, we have $E(Z(1=1)|A=0)\geq \frac{m_1}{BF_U}\geq \frac{m_1}{BF^*_U}$.

Second, we know that $\frac{E(Z(0)=1|A=1)}{m_0}=\frac{\int r_0(u)F_1(du)}{\int r_0(u)F_0(du)}=CMR_{AZ^+}\leq BF_U\leq BF^*_U$, then we have $E(Z(0)=1|A=1)\leq m_0 BF_U\leq m_0 BF^*_U$.

By definition of ACE and the inequalities derived above, we have
\begin{align*}
ACE_{AZ^+}^{true}&=m_1-\int r_0(u)F_1(du)\geq m_1-m_0\times BF^*_U ;\\
ACE_{AZ^-}^{true}&=\int r_1(u)F_0(du)-m_)\geq m_1/BF^*_U-m_0 ;\\
ACE_{AZ}^{true}&=f\cdot ACE_{AZ^+}^{true}+(1-f)ACE_{AZ^-}^{true}\\
&\geq f(m_1-m_0BF^*_U)+(1-f)(\frac{m_1}{BF^*_U}-m_0)\\
&=(m_1-m_0\times BF^*_U)[f+(1-f)/BF^*_U]\\
&=(\frac{m_1}{B^*F_U}-m_0)[f\times BF^*_U+(1-f)].
\end{align*}
Similarly, we can prove the inequalities hold for nonnegative outcomes and $ACE_{AZ}^{obs}<0$.
\end{proof}
\begin{proposition}
\label{sa_prop5}
In the matched sample, we have the following inequality for nonnegative outcomes and $ACE_{AZ}^{obs}\geq0$:
$$ACE_{AZ}^{true}\geq \frac{1}{2}(1+\frac{1}{BF_U^*})E(Z|A=1)-\frac{1}{2}(1+BF_U^*)E(Z|A=0).$$
In the matched sample, we have the following inequality for nonnegative outcomes and $ACE_{AZ}^{obs}<0$:
$$ACE_{AZ}^{true}\leq \frac{1}{2}(1+BF_U^*)E(Z|A=1)-\frac{1}{2}(1+\frac{1}{BF_U^*})E(Z|A=0).$$
\end{proposition}
\begin{proof}
In the matched sample, we have $f=P(A=1|e(X)=e(x))=0.5$. 
For nonnegative outcomes and $ACE_{AZ}^{obs}\geq0$, we have $LHS=ACE_{AZ}^{true}=\underset{e(x)}{\sum}ACE_{AZ|e(X)=e(x)}^{true}P(e(X)=e(x))$ and
\begin{align*}
RHS=&\underset{e(x)}{\sum}(m_1-m_0 BF^*_U)(f+\frac{1-f}{BF^*_U})P(e(X)=e(x))\\
=&(\frac{1}{2}-\frac{1}{2}BF^*_U)\underset{e(x)}{\sum}m_1P(e(X)=e(x))+\frac{1}{BF^*_U}\underset{e(x)}{\sum}m_1P(e(X)=e(x))\\
&+(\frac{1}{2}-\frac{1}{2}BF^*_U)\underset{e(x)}{\sum}M_0P(e(X)=e(x))-\underset{e(x)}{\sum}m_0P(e(X)=e(x))\\
=&\frac{1}{2}(1+\frac{1}{BF^*_U})E(Z|A=1)-\frac{1}{2}(1+BF^*_U)E(Z|A=0).
\end{align*}
 Thus, we have $ACE_{AZ}^{true}\geq \frac{1}{2}(1+\frac{1}{BF_U^*})E(Z|A=1)-\frac{1}{2}(1+BF_U^*)E(Z|A=0)$.

For $ACE_{AZ}^{obs}<0$, we have $LHS=ACE_{AZ}^{true}=\underset{e(x)}{\sum}ACE_{AZ|e(X)=e(x)}^{true}P(e(X)=e(x))$ and
 \begin{align*}
RHS=&\underset{e(x)}{\sum}(m_1 BF^*_U-m_0)(f+\frac{1-f}{BF^*_U})P(e(X)=e(x))\\
=&\underset{e(x)}{\sum}(m_1 BF^*_U-m_0)(\frac{1}{2}+\frac{1}{2BF^*_U})P(e(X)=e(x))\\
=&\frac{1}{2}(1+\frac{1}{BF^*_U})[BF^*_U\underset{e(x)}{\sum}m_1P(e(X)=e(x))-\underset{e(x)}{\sum}m_0P(e(X)=e(x))]\\
=&\frac{1}{2}(1+BF_U^*)E(Z|A=1)-\frac{1}{2}(1+\frac{1}{BF_U^*})E(Z|A=0).
\end{align*}
Thus, we have $ACE_{AZ}^{true}\leq \frac{1}{2}(1+BF_U^*)E(Z|A=1)-\frac{1}{2}(1+\frac{1}{BF_U^*})E(Z|A=0)$.
\end{proof}

\section*{Web Appendix C: Estimation of $G_{ij}(u,v)$}

To compute the variance of RMSTs, one difficulty is to estimate the function $G_{ij}(u,v)$ based on data. Follow the notations in \citet{murray2000variance}, we need to transform the function $G_{ij}(u,v)$ into the counting process notation system. Suppose we have $n$ matched pairs, then let $i,j$ denote the groups and $k=1,\cdots, n$ denotes the $k$th pair. Let $U_{ik}$ be the censoring random variable corresponding to survival time $T_{ik}$, and the censored survival time is $X_{ik}=\min (T_{ik},U_{ik})$ with censoring status $\Delta_{ik}=I(T_{ik}<U_{ik})$. Then we have the following definitions:
\begin{itemize}
\item $Y_i(u)=\underset{k=1}{\overset{n}{\sum}}I(x_{ik}\geq u)$ and $Y_j(v)=\underset{k=1}{\overset{n}{\sum}}I(x_{jk}\geq v)$;
\item $Y_{ij}(u,v)=\underset{k=1}{\overset{n}{\sum}}I(x_{ik}\geq u,x_{jk}\geq v)$;
\item $dNi(u)=\underset{k=1}{\overset{n}{\sum}}I(u\leq x_{ik}<u+\Delta u, \Delta_{ik}=1)$, where $\Delta u\rightarrow 0$;
\item $dNj(v)=\underset{k=1}{\overset{n}{\sum}}I(v\leq x_{ik}<v+\Delta v, \Delta_{ik}=1)$, where $\Delta v\rightarrow 0$;
\item $dN_{ij}(u,v)=\underset{k=1}{\overset{n}{\sum}}I(u\leq x_{ik}< u+\Delta u,v\leq x_{jk}<v+\Delta v, \Delta_{ik}=1,\Delta_{jk}=1)$, where $\Delta u\rightarrow 0$ and $\Delta v\rightarrow 0$;
\item $dN_{ij}(u|v)=\underset{k=1}{\overset{n}{\sum}}I(u\leq x_{ik}<u+\Delta u, x_{jk}\geq v, \Delta_{ik}=1)$, where $\Delta u\rightarrow 0$;
\item $dN_{j|i}(v|u)=\underset{k=1}{\overset{n}{\sum}}I(v\leq x_{jk}<v+\Delta v, x_{ik}\geq u, \Delta_{jk}=1)$, where $\Delta v\rightarrow 0$;
\item The $\hat{G}_{ij}(u,v)$ could be estimated by the formula below, and we set $\Delta u=0$ and $\Delta v=0$ in real computation. The corresponding R code could be found in our supplementary materials.
\begin{align*}
\hat{G}_{ij}(u,v)=&n\frac{Y_{ij}(u,v)}{Y_i(u)Y_j(v)}\Bigg[\frac{dN_{ij}(u,v)}{Y_{ij}(u,v)}-\frac{dN_{ij}(u|v)dN_j(v)}{Y_{ij}(u,v)Y_j(v)}\\
&-\frac{dN_{j|i}(v|u)dN_i(u)}{Y_{ij}(u,v)Y_i(u)}+\frac{dN_i(u)dN_j(v)}{Y_i(u)Y_j(v)}\Bigg]
\end{align*}
\end{itemize}


\end{document}